%% file: main-eurosp-2022.tex
\documentclass[compsoc,conference,a4paper,10pt,times]{IEEEtran}
\IEEEoverridecommandlockouts
\usepackage{cite}
\usepackage{amsmath,amssymb,amsfonts}
\usepackage{algorithmic}
\usepackage{graphicx}
\usepackage{textcomp}
\usepackage{bmpsize}
\usepackage{xcolor}
\usepackage{lipsum}
\usepackage[colorlinks=true,urlcolor=black]{hyperref}

\usepackage{url}
\usepackage[normalem]{ulem}
\usepackage{amsthm}
\usepackage{mathtools}
\usepackage{color}
\usepackage{xspace}
\usepackage{bussproofs}
\usepackage{thm-restate}
\usepackage{thmtools}
\usepackage{cleveref}
\usepackage{paralist}
\usepackage{stmaryrd}
\usepackage{enumitem}
\usepackage{todonotes}
\setlist[itemize]{leftmargin=2em,labelsep=0.5em}

\usepackage{tikz}
\usetikzlibrary{fit}
\usetikzlibrary{matrix}
\usetikzlibrary{shapes.geometric}
\usetikzlibrary{positioning}

\renewcommand{\paragraph}[1]{
  \medskip

  {\bf #1.}}

\usepackage{listings} 
\input{lstsettings-proverif}

\DeclareMathOperator{\verifype}{\mathsf{verify\_pe}}
\DeclareMathOperator{\verifypp}{\mathsf{verify\_pp}}
\DeclareMathOperator{\makepe}{\mathsf{pe}}
\DeclareMathOperator{\makepp}{\mathsf{pp}}
\DeclareMathOperator{\cons}{\mathsf{cons}}
\DeclareMathOperator{\encrypt}{\mathsf{encr}}
\DeclareMathOperator{\enckey}{\mathsf{enc\_key}}
\DeclareMathOperator{\decrypt}{\mathsf{decr}}
\DeclareMathOperator{\verkey}{\mathsf{ver\_key}}
\DeclareMathOperator{\checksign}{\mathsf{checksign}}
\DeclareMathOperator{\sign}{\mathsf{sign}}
\DeclareMathOperator{\represents}{\mathsf{repr}}

\newcommand{\nil}{\mathrm{nil}}
\newcommand{\myright}{\mathrm{right}}
\newcommand{\myleft}{\mathrm{left}}

\declaretheorem{theorem}
\declaretheorem{lemma}
\declaretheorem{definition}
\declaretheorem{example}

\declaretheorem{invariant}

\declaretheoremstyle[spaceabove=6pt, spacebelow=6pt,  headfont=\itshape,  notefont=\mdseries, notebraces={(}{)},bodyfont=\normalfont,  postheadspace=1em,  qed=$\blacktriangleright$]{examplestyle}

\def\BibTeX{{\rm B\kern-.05em{\sc i\kern-.025em b}\kern-.08em
    T\kern-.1667em\lower.7ex\hbox{E}\kern-.125emX}}

\usepackage{ifthen}
\newboolean{longversion}
\setboolean{longversion}{true}

\begin{document}

\ifthenelse{\boolean{longversion}}{
  \title{Automatic verification of transparency protocols (extended version)}
}{
  \title{Automatic verification of transparency protocols}
}

\author{\IEEEauthorblockN{Vincent Cheval}
\IEEEauthorblockA{\textit{INRIA Paris} \\
France\\
vincent.cheval@inria.fr}
\and
\IEEEauthorblockN{Jos\'e Moreira}
\IEEEauthorblockA{Valory AG \\
Switzerland\\
jose.moreira.sanchez@valory.xyz}
\and
\IEEEauthorblockN{Mark Ryan}
\IEEEauthorblockA{University of Birmingham \\
United Kingdom\\
m.d.ryan@cs.bham.ac.uk}
}

\maketitle

\input{macro_commands}
\input{abstract}


\input{introduction}

\section{Background and related work}

\subsection{Transparency protocols}

Several recent protocols use a publicly-accessible append-only log
data structure to achieve a \emph{transparency property}. One of the
earliest and most widely deployed protocol in this category is
\emph{certificate transparency} \cite{CT}. The core idea of certificate
transparency is that certificates are accepted by browsers only if
they are accompanied by a proof that they are present in an
appropriate log. Insisting on certificates being in a
publicly-accessible log means that the existence of the certificate is
\emph{transparent}. It prevents situations in which corrupted CAs
issue rogue certificates without being noticed.
This idea has been generalised to define more ambitious public-key
infrastructures, such as ARPKI \cite{ARPKI} and DTKI \cite{DTKI},
which aim to make all the infrastructure parties behave transparently.

\paragraph{The log and its proof data}
As mentioned, transparency protocols rely on an append-only log. The log is not assumed to be trustworthy; rather, anyone
 can verify the data it outputs, and if this verification
succeeds, then the transparency property is upheld.  The \emph{expected
  behaviour} of the log may vary from protocol to protocol; here, we
give a generic example \cite{Ryan:ndss}. The log $L$ is
organised as an append-only Merkle tree. For our purposes, a Merkle
tree storing data $R_1,R_2,\dots,R_n$ is a binary tree whose leaves
(when considered in left-to-right order) store the data
$R_1,R_2,\dots,R_n$ and whose non-leaf nodes store $H(c_\ell, c_r)$
where $H$ is a hash function and $c_\ell$ and $c_r$ is the data stored
at the left and right child node respectively. The maintainer of $L$
runs three protocols:
\begin{itemize}
\item On request, it outputs the current value $h$ stored at $L$'s root
  (called the \emph{root tree hash} of $L$).
\item On input $R$, it outputs data which proves that $R$ is present
  in $L$ (or it outputs $\bot$ if that is not the case). This data consists of
  data stored in some of the nodes of $L$, and its size is $O(\log
  n)$.
\item On input $h_1$ and $h_2$, it outputs data which proves that $L$
  previously 
  had the root tree hash $h_1$, and subsequently the root tree hash
  $h_2$ (or it outputs $\bot$ if this is not the case).  This data also consists of
  data stored in some of the nodes of $L$, and its size is $O(\log
  n)$.
\end{itemize}

The log's behaviour is fully verifiable, and therefore there
is no  trust assumption on the log maintainer.
It  can be malicious, and still the security property is upheld.

\paragraph{Verification of transparency protocols}
Several previous papers have applied protocol verification techniques
to transparency protocols.  The papers on ARPKI and
DTKI \cite{ARPKI,DTKI} both use the Tamarin prover to prove some
security properties. However, both papers make the same huge
abstraction: they treat the log as a list. The DTKI paper acknowledges
that formalising and modeling the complex data structures of transparency protocols
is an unsolved problem. Certificate Transparency was proved using Tamarin in \cite{Kunemann:accountability} using a simplified setting, for example by modelling the log as a trusted global database shared between agents. Additionally, in that work, proofs of presence and proof of extension were not modeled; instead, they use placeholders in the shared trusted database to act as proofs. In contrast, we model proofs and their data structure precisely.

\subsection{Transparent decryption}\label{sec:transparent_decryption}

\emph{Transparent decryption} is a transparency protocol, aiming to
prevent certain decryptions from being performed stealthily; rather,
the decryption operation inevitably produces
evidence of the fact that the decryption has taken place.  This can be
used, for example, to support privacy: it can mean that a subject is
alerted to the fact that information about them has been
decrypted. Among other uses, transparent decryption has been proposed for accountable
execution of search warrants and data interception
\cite{kroll2014secure, nunez2019escrowed, Ryan:17}; data sharing between
organisations \cite{idan2020prshare}; in
vehicle and IoT data applications \cite{9627586}. A company has begun
building products using these ideas \cite{pad.tech}.
Transparent
decryption is an accountable algorithm in the sense of
\cite{kroll2015accountable}.

In transparent decryption, the decryption key is distributed among a
set of agents (called \emph{trustees}); they use their key share only
if the required transparency conditions have been satisfied.
Typically, the transparency condition can be formulated as the
presence of the decryption request in a transparency log
\cite{frankle2018practical}.

We present a minimal system for transparent decryption below.  The
system satisfies the basic security property for transparent
decryption, which we also detail below.

\paragraph{How it works} More formally, the system works as follows:
\begin{itemize}
\item \emph{Subjects} create ciphertexts using a public encryption key $ek$.
\item Shares $dk_1,\dots,dk_n$ of the decryption key are held by
  \emph{trustees} $T_1,\dots,T_n$. For example, this might be a threshold
  decryption system, so that any $m$ out of $n$ trustees are
  sufficient to decrypt.
\item A \emph{decryption requester} $G$ can request the decryption of
  a ciphertext. This involves recording the request in a \emph{log} $L$.
\item $L$ is organised as a Merkle tree. This means that the
  maintainer can issue data that demonstrates it's maintaining $L$ in an
  append-only fashion (see below).
\item \emph{Trustees} are automatic processes which accept ciphertexts and
  log data as input.  The log data attests that certain information has
  been placed in the log $L$, and that $L$ has been maintained
  append-only since it was last seen by the trustee. The trustees
  verify the log data. If (and only if)
  the data verifies correctly, a
  trustee will perform its part in decrypting the relevant ciphertext
  and output the result.
\item Subjects can try to inspect the log contents. If their attempt is successful, they will see from the log
  whether their ciphertexts have been decrypted or not. If they are
  not successful (for example, the data is inconsistent or their
  access is denied) then they should assume that their ciphertexts
  have been decrypted.
\end{itemize}

The trust assumption for transparent decryption is that the trustees
behave correctly. As mentioned, this means that they perform their
part of the decryption if, and only if, the verification of the proofs
in the input data
is successful.

\paragraph{Security property} We aim to prove the following
property:
\begin{quote}
  Suppose an honest subject encrypts a secret $s$ with $ek$, and later
  the secret $s$ becomes known by some other party (e.g., any of the
  mentioned parties, or an attacker, or anyone else). Suppose the
  subject successfully accesses the log $L$ and successfully verifies
  the log data. Then the subject sees the decryption request for $s$ in $L$.
\end{quote}

\paragraph{Trustees and their actions} 
Trustees are designed to be very simple and to have minimal
computational requirements, so that their trustworthiness can be
established as straightforwardly as possible. They do not have to
store any voluminous data; they store just three data items, and they
run two protocols (see Fig.~\ref{fig:trustee}).

\begin{figure}
  \centering
  \begin{tabular}{|c|c|}
    \hline
    \multicolumn{2}{|c|}{
    \begin{minipage}{0.85\linewidth}
      \ \\[-2ex]
      \begin{center}
        $T_i$ stores: $h$, $dk_i$, $sk_i$  \\\
      \end{center}
    \end{minipage}
    }
    \\\hline
    \begin{minipage}[t]{0.48\linewidth}
      \ \\
      \mbox{}\ $\bullet$\quad Input: $R, h', \pi, \rho$\\
      \mbox{}\ $\bullet$\quad Compute:\\
      \mbox{}\ \quad --  Verify $\pi$: $R$ in $h'$\\
      \mbox{}\ \quad --  Verify $\rho$: $h'$ extends $h$\\
      \mbox{}\ \quad --  result := dec($dk_i, R$)\\
      \mbox{}\ \quad --  $h$ := $h'$\\
      \mbox{}\ $\bullet$\quad Output: result\\\
    \end{minipage} &
    \begin{minipage}[t]{0.45\linewidth}
      \ \\
       \mbox{}\ $\bullet$\quad  Input: $v$\\
       \mbox{}\ $\bullet$\quad  Compute:\\
       \mbox{}\hspace{1.7em}r := sign$(sk_i, (v,h))$\\
      \mbox{}\ $\bullet$\quad  Output: r
    \end{minipage}\\\hline
  \end{tabular}
  \caption{Protocols run by trustee $T_i$. The trustee stores the most recent
    root tree hash $h$ of the log that it has seen, and a
    decryption key $dk_i$ and a signing key $sk_i$ share. The protocol on
    the left inputs a request $R$ and some other parameters, and
    outputs a decrypted result. The protocol on the right inputs a
    nonce $v$, and outputs a signature on $(v, h)$.}
  \label{fig:trustee}
\end{figure}

Trustees can be implemented in a variety of ways. For example, they
may be cloud-based software processes run by organisations with a high reputation such as
charities and foundations. These organisations can use hardware-based
attestation to give further confidence about the binary code trustees
are running, and the secure storage and use of their
keys. Alternatively, trustees could be implemented on dedicated
hardware modules, such as the TPM \cite{proudler2014trusted} or Google
Titan chip \cite{johnson2018titan}, or a RISC-V
chip like Open Titan \cite{opentitan,moller2021preliminary}.

The system we have described is a minimal one that provides decryption
transparency. It could readily be extended to have some additional
properties, such as trustee obvliviousness
(namely, the inability of a trustee to obtain any information about
the decryption request or its result), and proper authentication of
the decryption request (see, e.g. \cite{kroll2014secure}).

\paragraph{Applications} Transparent decryption can be applied in many
areas to enhance privacy. We give some examples.
\begin{itemize}
\item Alex can choose to share her location in encrypted form with
  some nominated friends and family, called \emph{angels} by the app
  that implements this idea \cite{pad-places-app}. No-one except her
  angels can view her location; and she can monitor whether and when
  they do so.
\item Suppose Alex is being investigated by the police. In an effort
  to establish her innocence, she may choose to hand over her phone.
   With transparent decryption, Alex can upload her phone contents in encrypted form.
   Then Alex gets evidence of what part of this
   uploaded material is decrypted.
\item Alex provides \emph{know-your-customer} (KYC) information to her bank, so that if necessary
   later, it can carry out anti-money laundering procedures.
   Recent proposals \cite{parra2017kyc} suggest centralising KYC
   registers, to make the procedure more efficient. With transparent
   decryption, the ill effects of such centralisation can be mitigated
   by making money laundering investigations more transparent.
\end{itemize}

\subsection{ProVerif} 
ProVerif is a software tool for automated reasoning about the security
properties of cryptographic protocols.  It was first released in 2002,
and has been continuously developed for the last 20 years. 
It has been
used to analyze hundreds of protocols, including major
deployed protocols such as TLS \cite{c4}, Signal \cite{c25}, Noise \cite{c26},
avionic protocols \cite{c9}, and the Neuch\^atel
voting protocol \cite{c17}. 

   A cryptographic protocol in ProVerif is specified as follows:
\begin{itemize}
\item Cryptographic primitives (such as symmetric and asymmetric encryption or digital signatures) are specified typically as \emph{reduction rules}, such as this one for public key encryption and decryption: $\decrypt(k, \encrypt(\enckey(k), m)) = m$.
\item The behaviour of the protocol participants is described using the process calculus syntax (see below).
\item The properties which are to be checked are specified as queries. ProVerif supports different kinds of properties; in this paper, we restrict our attention to reachability and correspondence properties. For example, the correspondence property $\fevent{ev(x)} \Rightarrow \fatt{x}$ says that if the event $ev$ occurs with a parameter value $x$, then the value $x$ was previously known by the attacker.  
\end{itemize}

\paragraph{Syntax} A simplified syntax for the process calculus \emph{terms}, \emph{expressions}, \emph{events}, \emph{predicates} and \emph{processes} is displayed in Fig.~\ref{fig:syntax}. ProVerif's calculus also supports additional constructs, e.g. for tables, phases, extended terms, \ldots, but we omit them for simplicity as our results can be easily generalized to these constructs. 

Terms $M, N, \dots$ are built over variables, names and application of \emph{constructor} function symbols from a finite set $\Fc$. \emph{Destructor} function symbols, from a finite set $\Fd$, can manipulate terms and must be evaluated in the \emph{assignment} construct. Unlike constructor function symbols, the evaluation of a destructor may \emph{fail}, or in other words, may evaluate to the special constant $\fail$. Typically, in the assignment construct $\llet x=D \inelse P \eelse Q$, the expression $D$ will be evaluated; if its evaluation fails then the process $Q$ will be executed, otherwise the variable will be instantiated by its result and $P$ will be executed. The exact behavior of a destructor function symbol is defined by a list of rewrite rules given by the user (see~\cite{blanchet:hal-03366962} for the complete definition of the evaluation of an expression).
For example,
\[
\inn(c, y); \llet x = \decrypt(k, y) \inelse \out(c, x) \eelse \out(c, 0)
\]
outputs the plaintext $m$ if a ciphertext of the form $\encrypt(\enckey(k), m)$ is given
as input, otherwise it outputs~0.
\begin{figure}[ht]
\begin{center}
\begin{minipage}{\columnwidth}
  \begin{tabbing}
    $P,Q$  \= ::= \hspace{3.5cm}\=   pro\=cesses\\
    \>$0$\>\> nil\\
    \>$\out(N,M);P$\>\>output\\
    \>$\inn(N,x);P$\>\>input\\
    \>$P \mid Q$\>\>parallel composition\\
    \>$!P$\>\>replication\\
    \>$\new a; P$\>\>restriction\\
    \>$\llet x=D \inelse P \eelse Q$\>\>assignment\\
    \>$\llet x_1,\ldots, x_n \suchthat p(M_1,\ldots,M_k) \inelse P \eelse Q$\\
    \> \>\>predicate evaluation\\
    \>$\event(ev(M_1,\ldots,M_k)); P$\>\>event
  \end{tabbing}
  \end{minipage}
\end{center}
\caption{Syntax of the core language of ProVerif.}
\label{fig:syntax}
\end{figure}

A substitution $\sigma$ is an assignment of terms to some variables;
for example, $\{x\mapsto \encrypt(k, m)\}$ is a
substitution.  If $M$ is a term, then $M\sigma$ is the term obtained
by replacing any $x$ mapped by the substitution with the term that it
maps to. For processes $P$ and facts $F$, applying the substitution to
obtain $P\sigma$ and $F\sigma$ is
defined similarly (taking care not to substitute bound variables).

The process calculus also contains standard constructs
$\out(N,M); P$ (representing the output of a term $M$ on a channel $N$),
$\inn(N,x); P$ (the input on channel $N$ of a message which gets bound to a variable $x$),
$\new a; P$ (the generation of a fresh name $a$),
$P \mid Q$ (the concurrent execution of processes),
$\event(ev(M_1,\ldots,M_k)); P$ (the recording of event execution), and
$!P$ (the concurrent execution of an unbounded number of copies of a process).

Less common is the construct for predicate evaluation, that is, $\llet x_1,\ldots, x_n \suchthat p(M_1,\ldots,\allowbreak M_k) \inelse P \eelse Q$. In this construct, the variables $x_1,\ldots, x_n$ must occur in the predicate $p(M_1,\ldots,M_k)$. If $x_1,\ldots, x_n$ can be instantiated, say by a substitution $\sigma$, such that $p(M_1,\ldots,M_k)\sigma$ holds then $P\sigma$ is executed; otherwise $Q$ is executed. Note that predicate evaluations are most commonly used with a classical if-then-else conditional, corresponding in fact to the evaluation without variables ($n = 0$).

\paragraph{User-defined predicates and clauses}
ProVerif allows users to define predicates, and to give their
semantics by means of Horn clauses.  This is useful for defining predicates on data types. For example, the list data structure can be represented by a constant $\nil$ and a constructor $\cons$. One can define a membership predicate $\mem$ using the Horn clauses:
\begin{align*}
\forall x,\ell.\ &\mem(x,\cons(x,\ell))\\
\forall x,y,\ell.\ &\mem(x,\ell) \rightarrow \mem(x,\cons(y,\ell)).
\end{align*}
As variables in Horn clauses are always universally quantified, we will omit writing the quantifier in the rest of this paper. 


\begin{definition}[derivation]
A derivation $\Deriv$ of a fact $F$ from a set of clause $\Cuser$ is a tree whose nodes are labeled by Horn clauses in $\Cuser$ and edges are labeled by ground facts such that the incoming edge of the root is labeled by $F$.  For all nodes $\eta$ in $\Deriv$ labelled by a clause $F_1 \wedge \ldots \wedge F_n \wedge \phi \rightarrow C$, there exists a substitution $\sigma$ such that:
\begin{inparaenum}[(i)]
	\item the incoming edge of $\eta$ is labeled by $C\sigma$;
	\item $\eta$ has $n$ outgoing edges labeled by $F_1\sigma, \ldots, F_n\sigma$ respectively;
	\item $\phi\sigma$ is true.
\end{inparaenum}
\end{definition}

The derivability of facts allows us to define the true statements of a predicate $p$, denoted $\sem{p}$, as the set of facts $F = p(M_1,\ldots, M_n)$ derivable from $\Cuser$.  Note that only user-defined predicates, equalities and disequalities on terms can occur in the clauses from $\Cuser$. As such, the semantics of a user-defined predicate is independent from any protocol.

\begin{example}
$\mem(a,\cons(b,\cons(a,\nil))) \in \sem{\mem}$ as it is derivable by the following derivation with $\sigma_1 = \{ x \mapsto a; y \mapsto b; \ell \mapsto \cons(a,\nil)\}$ and $\sigma_2 = \{ x \mapsto a; \ell \mapsto \nil\}$.
\begin{center}
\vspace{-0.7cm}
\begin{tikzpicture}[
	derivation/.style = {draw,regular polygon, regular polygon sides=3,inner sep = -0.1cm },
	rule/.style = {draw, rounded corners},
	auto
	]
	
	\node (root) {};
	\node[rule] (R1) [below=0.5cm of root,label=right:{\ with $\sigma_1$}] {$\mem(x,\ell) \rightarrow \mem(x,\cons(y,\ell))$};
	\node[rule] (R2) [below=0.5cm of R1,label=right:{\ with $\sigma_2$}] {$ \mem(x,\cons(x,\ell))$};
		
	\path[<-] 
		(R1) edge node[auto] {$\mem(a,\cons(b,\cons(a,\nil)))$} (root)
		(R2) edge node[auto] {$\mem(a,\cons(a,\nil))$} (R1);
\end{tikzpicture}
\end{center}
\end{example}

Optionally, a predicate can be declared as a ``blocking'' predicate, meaning that there are no clauses containing the predicate in the conclusion of the clause. In this case, ProVerif proves properties that hold for \emph{any} definition of the
considered blocking predicate.

\paragraph{Semantics of processes}
\newcommand{\lstep}[1]{\xrightarrow{#1}}
The semantics of processes is defined by the means of a reduction relation $\lstep{\ell}$ between \emph{configurations} which express the current state of the the execution of the processes interacting with the attacker. Formally, a configuration is a triple $\E,\Pro,\Att$ where $\E$ is the set of names used in the configuration, $\Pro$ is a multiset of processes, and $\Att$ is a set of terms representing the knowledge of the attacker.

\ifthenelse{\boolean{longversion}}{
The full set of rules defining the relation $\lstep{\ell}$ is provided in Appendix (and in~\cite{ProVerifManual}) and we only show a small extract below. 
}{
The full set of rules defining the relation $\lstep{\ell}$ is provided in~\cite{tech} (and in~\cite{ProVerifManual}) and we only show a small extract below. 
}
For example, the following rule represents that an event is triggered:
\[
\E,\Pro \cup \multiset{\event(ev).P} ,\Att \lstep{\fevent{ev}} \E, \Pro \cup \multiset{P}, \Att
\]
The rule for predicate evaluation (when the predicate evaluates to true) is:
\begin{multline*}
\E, \Pro \cup \multiset{\llet x_1,\ldots, x_n \suchthat pred \inelse P \eelse Q} \\ \lstep{}  \E,\Pro \cup \{P\sigma\}, \Att
\end{multline*}
when $pred = p(M_1,\ldots,M_k)$ and there exists a substitution
$\sigma$ such that $\dom{\sigma} = \{ x_1,\ldots,x_n\}$ and
$p(M_1\sigma,\ldots,M_k\sigma) \in \sem{p}$.

An \emph{execution trace} of a process $P$ is then defined as a
sequence of applications of the relation $\lstep{\ell}$ starting from
the initial configuration $\mathcal{C}_1 =
(\emptyset,\multiset{P},\emptyset)$, i.e. $T = \mathcal{C}_1
\lstep{\ell_1} \ldots \lstep{\ell_n}
\mathcal{C}_{n+1}$. 

In addition to user-defined predicates, ProVerif considers several native predicates: $\fatt{M}$ indicating that the attacker knows $M$; $\fmsg{M}{N}$ indicating that a message $N$ has been sent on the channel $M$ and $\fevent{ev}$ indicating that an event $ev$ has been raised. Satisfiability of a fact $F$ by a trace $T$,
denoted $T \satisfy F$, is given by the labels and configurations in
$T$ (e.g. when $F = \fevent{ev} = \ell_i$ for some $i$).
 Naturally, $T \satisfy F$ with $\pred(F) \in \Fp$ when $F \in \sem{p}$.

Finally, a correspondence query $F_1 \wedge \ldots \wedge F_n\Rightarrow \psi$ can be seen as the first order logic formula $\Psi = \forall \tilde{x}. (F_1 \wedge \ldots \wedge F_n \Rightarrow \exists \tilde{y}. \psi)$ where $\tilde{x} = \vars{F_1,\ldots,F_n}$ and $\tilde{y} = \vars{\psi} \setminus \tilde{x}$. The correspondence query holds when for all traces $T$ of $P$, $T \satisfy \Psi$.

\paragraph{Lemmas and axioms} 
The problem that ProVerif tries to solve is undecidable in general \cite{abadi06}; therefore, by design
ProVerif is not complete: it may fail to terminate, and it may yield false attacks. Much work has been done to make it more complete in practice. A significant step in this direction  introduces \emph{lemmas and axioms} \cite{blanchet:hal-03366962} as a way to guide derivations in ProVerif, and also to deal with some of the abstractions introduced by the tool.



An \emph{axiom} is an instruction to ProVerif to consider some facts as true, even if they cannot be proved by ProVerif from the protocol process. Typically, an axiom is used if one has a separate (perhaps manual) proof of the fact in question. Consider, for example, a smartcard which stores two secrets, $s_1$ and $s_2$. It allows the user to choose either one of them to be revealed, but not both. (This might be used in a lottery, for example.) We could model the smartcard with the following process:
\[
\inn(c, x); ( \iif x=1 \then \out(c, s_1) \mid \iif x=2 \then \out(c, s_2) )
\]

The user chooses to enter 1 or 2, and obtains the corresponding secret; after that, the smartcard does not accept any further input.

The intended security property is that at most one secret is revealed. Unfortunately, ProVerif is not able to prove the security of this device with the given process description. The reason is that ProVerif introduces an abstraction, which allows it to consider a derivation in which $x$ has sometimes the value 1 and sometimes the value 2, and hence the output of $s_1$ and $s_2$ can both occur. This is not a valid trace; it is a false attack introduced by ProVerif's abstraction.

Axioms allow us to rectify this situation. We write the process as follows:
\begin{align*}
    &\new st; \inn(c, x); \event(\mathrm{Uniq}(st, x)); \\
    &\quad( \iif x=1 \then \out(c, s_1) \mid \iif x=2 \then \out(c, s_2) )
\end{align*}
and we add the axiom
\[
\mathrm{Uniq}(st, x_1) \wedge \mathrm{Uniq}(st, x_2) \Rightarrow x_1 = x_2.
\]
The axiom asserts that only one value of $x$ is allowed. This axiom is valid, since for a given $st$ there can be only a single input of $x$. ProVerif is able to use this axiom to prove the security of the device. This kind of axiom, stating that only one value of an input is allowed, is very useful in increasing the precision of ProVerif.

A \emph{lemma} is similar to an axiom as ProVerif uses it in proofs to establish the desired property but it must be able to prove it first
(while it does not try to prove axioms). Lemmas are therefore a useful way of decomposing a verification into smaller pieces.

\section{A methodology to model protocols with complex data structures}
An intuitive, first attempt to model the transparent decryption protocol in ProVerif
requires one to define: 
\begin{itemize}
	\item The required equational theories (e.g. public key encryption).
	\item The predicates and clauses defining the data structure for the log maintainer.
	\item The predicates and clauses that represent true statements about the proof of presence and extension that the data structure must satisfy.
	\item The process defining the protocol.
	\item The security properties that we are interested to test.
\end{itemize}

Unfortunately, this ``monolithic'' attempt to prove the security
properties does not always work when the protocol has to deal with
complex or recursive data structures, and it will often result in
ProVerif failing to terminate. In fact, we have encountered this issue even when we try this approach by substituting Merkle trees with a simpler data structure like a hash list.

\begin{figure}
	\centering
	\includegraphics[width=0.9\linewidth]{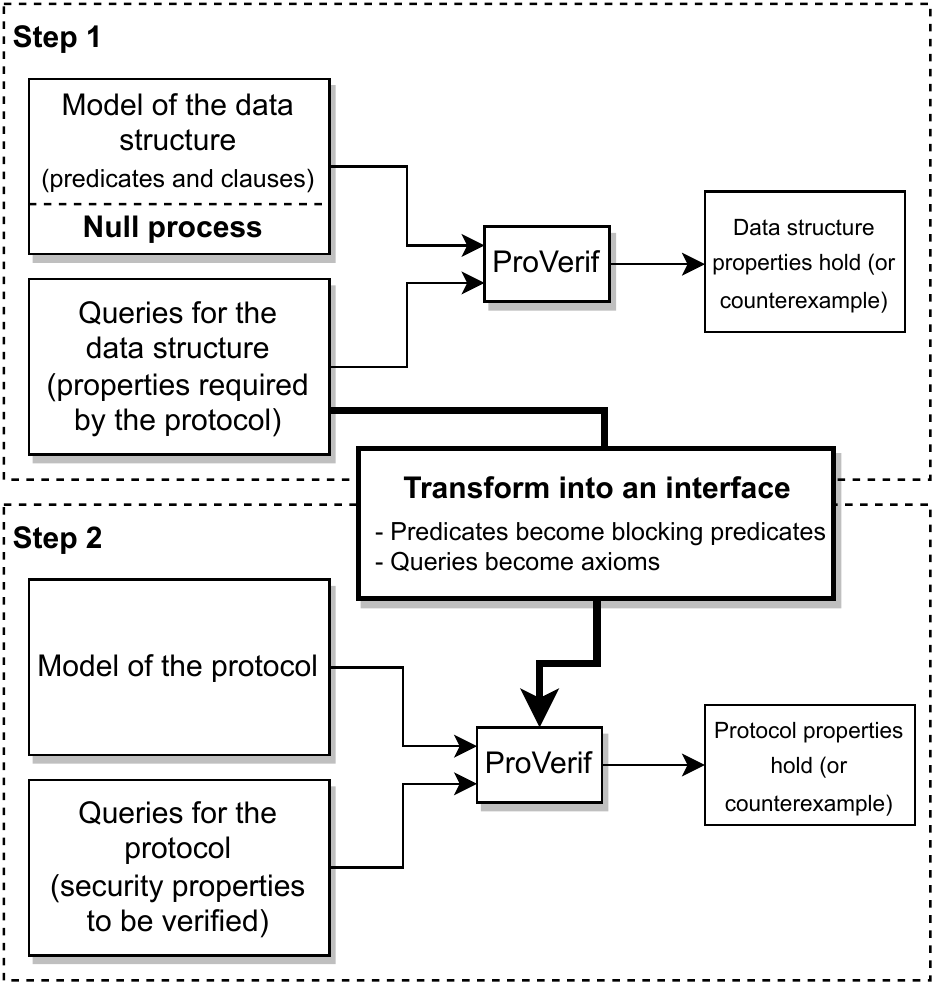}
	\caption{A methodology to prove security in protocols involving complex data structures}
	\label{fig:methodology}
\end{figure}

One way to avoid the monolithic proof is to make an appropriate abstraction of the Merkle tree data structure. 
This could be achieved by defining the properties the data structure is expected to satisfy, and proving the protocol based on assuming those properties. Later, as a separate step, one can look at whether a particular implementation satisfies the assumed properties.

This observation suggests the following generic methodology to approach proving security when a protocol requires usage of complex data structures. The key idea is to decompose the proof, separately proving security for the data structure and for the protocol separately, using these steps:
\begin{enumerate}[label={S\arabic*:}, ref={S\arabic*}] 
	\item \label{step:method1} Identify and extract the properties of the data structure that are required by the protocol, and prove that there is an implementation of a data structure that satisfies these properties. As the semantics of the data structure are protocol-independent, so are these properties.
	\item \label{step:method2} Assuming that we have a data structure with such properties, prove security for the protocol itself.
\end{enumerate}

We depict this methodology in Fig.~\ref{fig:methodology}. For the case of accountable decryption it works as follows. In Step~\ref{step:method1}, we propose a suitable model for the data structures employed by the ledger, which we instantiate either as a hash list or as a Merkle tree. We also define the desired properties of the data structure from the point of view of the protocol. The reason to consider two data structures is to show that the identified ``interface of security properties'' and the methodology are sufficiently generic to allocate several data structures and be considered for several protocols. Next, we use ProVerif to prove that the data structure satisfies the properties in the null process.

On the other hand, in Step~\ref{step:method2}, we adopt the properties for the data structure as an interface. This requires to take two actions: transform these properties as ProVerif axioms, and transform any related predicate to a blocking predicate. This will allow us to prove the security property of the protocol disregarding the particular implementation of the data structure. Observe that this step is independent from Step~\ref{step:method1} above, i.e. they both can be executed in parallel. 

By decomposing the problem into two parts, ProVerif has a better chance of avoiding nontermination, without losing soundness of the proof. Moreover the ``interface'' approach allows for proving security for more complex data structures than Merkle trees, and allows proving security of the protocol for any underlying data structure satisfying the data structure properties.

The remainder of this section describes in detail the methodology, by detailing the implementation of the two data structures commented above, the interface of security properties of the data structure, and the modelling and analysis of the accountable decryption protocol.

\subsection{Modeling the ledger data structure} \label{sec:datastructure}
Recall that the ledger is an untrusted party that maintains an append-only log $L$.
From the point of view of proving security properties for a protocol, it does not matter which is the particular data structure employed by the log maintainer, as long as it is append-only, and provides an interface to construct and verify proofs of presence and proofs of extension. Whichever data structure we use must define the clauses from which valid proofs of presence and the proofs of extension can be derived. We discuss two common data structures below, namely hash lists and Merkle trees.

Even though the most  interesting case is indeed the latter one, the main ideas of our approach are more clearly seen using hash lists. For this reason, we provide a more detailed explanation on how to model the proofs on hash lists, whereas we provide a brief example on Merkle trees, avoiding routine technicalities that might hinder the main points we want to state. The complete models, both for hash lists and Merkle trees are available at~\cite{supplement}.

For convenience, we parameterize the hash function $H$ with two values, e.g. $H(x,h)$, where $h$ is an output of the hash function. This can be implemented simply as concatenating the values $x$ and $h$ in a regular hash function. Hashes of single values $x$ are interpreted as $H(x, h_0)$, where $h_0$ denotes the null hash.

\subsubsection{Hash lists}
First, consider the case where $L$ is a hash list. That is, $L$ is represented by $h$, where
\[
h= H(R_n, \dots, H(R_2, H(R_1, h_0))\cdots).
\]
In this case, a proof of presence of $R_i$ in the list represented by $h$ simply consists of the elements inserted after $R_i$, 
plus the hash of the list before inserting $R_i$, i.e.
\begin{align*}
\pi = \makepp\bigl(
	&(R_n, \dots, R_{i+1}),h'\bigr),
\end{align*}
where  $h'=H(R_{i-1},\dots, H(R_1, h_0)\ldots)$ and $\makepp$ is  the constructor for the proof of presence.
The predicate $\verifypp(\pi, R_i, h)$ states that $\pi$ is a valid proof of presence of $R_i$ in the list represented by $h$. For a hash list,
$\verifypp$ will  check that the following equation holds:
\begin{equation}
h = H(R_n, \dots, H(R_{i-1}, H(R_i, h'))\cdots).\label{eq:verifypp_impl}
\end{equation}

In order for ProVerif to handle predicates in the resolution algorithm, we need to provide the
set of Horn clauses from which all valid predicates can be derived inductively. In the case of the proof of presence, these clauses
are
\begin{align}
	&\verifypp(\makepp(\nil, h'), R, H(R, h')),\label{eq:ppclause1}\\
	&\verifypp(\makepp(\ell, h'), R, h) \rightarrow\notag\\
	&\quad\verifypp(\makepp(\cons(Q, \ell), h'), R, H(Q,h)). \label{eq:ppclause2}
\end{align}

Clause~\eqref{eq:ppclause1} states the ``base case'': a proof of presence to verify the last entry on a hash list is  an empty list of elements and the immediate hash $h'$ before inserting $R$. It can be readily seen that the verification from Eq.~\eqref{eq:verifypp_impl} holds.
Clause~\eqref{eq:ppclause2} states the recursive nature of the proof of presence: if a list contains $R$, an extension of this list with an element $Q$ also contains $R$, and a proof of presence can be easily derived by prepending $Q$ into the list of elements of the original proof of presence.

For the proof of extension, let $L_1$ and $L_2$ be two hash lists represented by $h_1$ and $h_2$, respectively, and with lengths
$n_1\leq n_2$.
A proof of extension $\rho$ that $L_2$ extends $L_1$ simply consists of the list of elements inserted into $L_2$ after the
last element $R_{n_1}$ inserted into $L_1$, that is, 
\[
\rho = \makepe(R_{n_2},\dots, R_{n_1+1}),
\]
where $\makepe$ is the constructor for proofs of extension. The implementation of predicate $\verifype(\rho, h_1, h_2)$,
which verifies that $\rho$ is a valid proof of extension for the lists $L_1, L_2$, will check that the equation below holds:
\begin{equation}
	h_2 = H(R_{n_2}, \dots, H(R_{n_1+1}, h_1)\cdots).\label{eq:verifype_impl}
\end{equation}

The predicate $\verifype$ is defined by the following Horn clauses:
\begin{align}
&\verifype(\makepe(\nil), h, h),\label{eq:peclause1}\\
&\verifype(\makepe(\ell),h_1, h_2) \rightarrow\notag\\
&\quad\verifype(\makepe(\cons(R,\ell)), h_1, H(R, h_2)).\label{eq:peclause2}
\end{align}
Indeed, Clause~\eqref{eq:peclause1} above is the ``base case,'' and indicates that a hash list $L$ represented by $h$ extends itself trivially, hence the proof of extension is the empty list. Clause~\eqref{eq:peclause2} states that given a valid proof of extension indicting that a list $L_2$ represented by $h_2$ extends $L_1$ represented by $h_2$, then $\makepe(\cons(R,\ell))$ is a valid proof of extension stating that $R$ appended to $L_2$ also extends $L_1$.

Note that both $\pi$ and $\rho$ are data structures of size $O(n)$, and the verification of these proofs (Eq.~\eqref{eq:verifypp_impl} and~\eqref{eq:verifype_impl}) also takes time $O(n)$.
 
Finally, we also require a predicate $\represents(\ell, h)$
to state the fact that $h$ represents a hash list data structure containing the elements in $\ell=(R_n,\dots, R_1)$, inserted in reverse order. Clearly, the clauses defining this predicate are
\begin{align*}
  &\represents(\nil, h_0),\\
  &\represents(\ell, h) \rightarrow \represents(\cons(R, \ell), H(R, h)).
\end{align*}

\subsubsection{Merkle trees}
\begin{figure}
\begin{center}
\begin{tikzpicture}[
	derivation/.style = {draw,regular polygon, regular polygon sides=3,inner sep = -0.15cm },
	rule/.style = {draw, rounded corners},
	PP/.style = {draw=blue, dashed,rounded corners},
	PE/.style = {draw=red,rounded corners},
	auto
	]

	\node[rule] (root) {$h_{17}$};
	
	\node (titleT) [above=0.2cm of root] {Merkle tree $T_2$};
	
	\node[rule] (R11) [below left=0.5cm and 0.9cm of root] {$h_{14}$};
	\node[rule] (R12) [below right=0.5cm and 0.9cm of root] {$h_{57}$};
	\node[PP] (R21PP) [inner sep = 0.02cm,fit=(R12)] {};
	\node[PE] (R21PE) [inner sep = 0.05cm,fit=(R12)] {};
	
	\node[rule] (R21) [below left=0.5cm and 0cm of R11] {$h_{12}$};
	\node[PP] (R21PP) [inner sep = 0.02cm,fit=(R21)] {};
	\node[PE] (R21PE) [inner sep = 0.05cm,fit=(R21)] {};
	
	\node[rule] (R22) [below right=0.5cm and 0cm of R11] {$h_{34}$};
	
	\node[rule] (R31) [below left=0.5cm and -0.3cm of R21,label=below:{$R_1$}] {$h_1$};
	\node[rule] (R32) [below right=0.5cm and -0.3cm of R21,label=below:{$R_2$}] {$h_2$};
	
	\node[rule] (R41) [below left=0.5cm and -0.3cm of R22,label=below:{$R_3$}] {$h_3$};
	\node[PP] (R41PP) [inner sep = 0.02cm,fit=(R41)] {};
	
	\node[rule] (R42) [below right=0.5cm and -0.3cm of R22,label=below:{$R_4$}] {$h_4$};
	\node[PE] (R41PE) [inner sep = 0.02cm,fit=(R42)] {};
	
	\node[rule] (R51) [below left=0.5cm and -0.3cm of R12] {$h_{56}$};
	\node[rule] (R52) [below right=0.5cm and -0.3cm of R12,label=below:{$R_7$}] {$h_7$};
	
	\node[rule] (R61) [below left=0.5cm and -0.3cm of R51,label=below:{$R_5$}] {$h_5$};
	\node[rule] (R62) [below right=0.5cm and -0.3cm of R51,label=below:{$R_6$}] {$h_6$};
	
	\path[-] 
		(root) edge node[auto,text=blue,swap,pos=0.8] {$\myleft$} (R11)
		(root) edge node[auto,text=red,swap,pos=0.2] {$\myleft$} (R11)
		(root) edge node[auto] {} (R12)
		(R11) edge node[auto] {} (R21)
		(R11) edge node[auto,text=blue,pos=1] {$\myright$} (R22)
		(R11) edge node[auto,text=red,pos=0.4] {$\myright$} (R22)
		(R21) edge node[auto] {} (R31)
		(R21) edge node[auto] {} (R32)
		(R22) edge node[auto,text=red,pos=0.8,swap] {$\myleft$} (R41)
		(R22) edge node[auto,text=blue,pos=0.8] {$\myright$} (R42)
		(R12) edge node[auto] {} (R51)
		(R12) edge node[auto] {} (R52)
		(R51) edge node[auto] {} (R61)
		(R51) edge node[auto] {} (R62)
		;

	\node[rule] (root') [left=3.3cm of root] {$h_{13}$};
	
	\node (titleT') [above=0.2cm of root'] {Merkle tree $T_1$};
	
	\node[rule] (R11') [below left=0.5cm and -0.1cm of root'] {$h_{12}$};
	\node[rule] (R12') [below right=0.5cm and -0.1cm of root',label=below:{$R_3$}] {$h_3$};
	
	\node[rule] (R21') [below left=0.5cm and -0.25cm of R11',label=below:{$R_1$}] {$h_1$};
	\node[rule] (R22') [below right=0.5cm and -0.25cm of R11',label=below:{$R_2$}] {$h_2$};
	
	\path[-] 
		(root') edge node[auto,text=blue,swap] {} (R11')
		(root') edge node[auto] {} (R12')
		(R11') edge node[auto] {} (R21')
		(R11') edge node[auto,text=blue] {} (R22')
		;
\end{tikzpicture}

\end{center}

\raggedright Label $h_{ij}$ is the digest containing $R_i,\ldots,R_j$.\\ The proof of presence of $R_4$ in $T_2$ is  $\bigl((\myleft,h_{57}), (\myright,h_{12}), (\myright,h_{3}))\bigr)$, shown in blue. The proof
of extension of $T_1$ to $T_2$ is $R_3, \bigl((\myleft,h_{57}), (\myright,h_{12}), (\myleft,h_{4}))\bigr)$, shown in red.

\caption{Examples of Merkle trees}
\label{fig:merkletree}
\end{figure}

We assume that the reader has some familiarity with Merkle trees.
It is a more efficient data structure to construct the proofs: the memory and time requirements are reduced to $O(\log n)$ (compared with $O(n)$ for hash lists).
We omit most of the formalities and consider the example of Merkle trees depicted in Fig.~\ref{fig:merkletree}. Recall that a Merkle tree is a binary tree that assigns a hash value for each node, computed as the hash 
of its child nodes, e.g. $h_{12} = H(h_{1}, h_{2})$. By convention, the hashes on the leaf nodes are defined as the hash of the associated element, i.e.  $h_i=H(R_i,h_0)$. Hence, the root tree hash that represents the tree $T_2$ in Fig.~\ref{fig:merkletree} is $h=h_{17}=H(h_{14}, h_{57})$.

A proof of presence that $R_i$ is present in the tree is constructed by providing
the complementary node at each tree level. For example, the proof of presence $\pi$ of $R_4$ in $T_2$ is the list of values
\[
\pi = \makepp\bigl((\myleft,h_{57}), (\myright,h_{12}), (\myright,h_{3})\bigr).
\]
The constants ``$\myleft$'' and ``$\myright$'' are appended to each element to record the relative position of the path leading to the root three hash. Then, $\verifypp(\pi, R_4, h_{17})$ will check that
\[
h = H(H(h_{12}, H(h_{3}, H(R_4, h_0))), h_{57}).
\]

The predicate $\verifypp$ is defined by the following Horn clauses:
\begin{align}
&\verifypp(\makepp(\nil), R, H(R, h_0)),\label{eq:mt_ppclause1}\\
&\verifypp(\makepp(\ell), R, h_\ell) \rightarrow \notag\\
&\quad \verifypp\bigl(\makepp\bigl(\cons( (\myleft, h_r), \ell)\bigr), R, H(h_\ell, h_r)\bigr),\label{eq:mt_ppclause2}\\
&\verifypp(\makepp(\ell), R, h_r) \rightarrow \notag \\
&\quad\verifypp\bigl(\makepp\bigl(\cons((\myright, h_\ell), \ell)\bigr), R, H(h_\ell, h_r)\bigr).\label{eq:mt_ppclause3}
\end{align}
It is not difficult to see that these clauses represent inductively the proof of presence, starting from the base case in Clause~\eqref{eq:mt_ppclause1}, and defining recursively the proofs whose next step are left or right paths in Clauses~\eqref{eq:mt_ppclause2} and~\eqref{eq:mt_ppclause3}, respectively.

A proof of extension  for Merkle trees can be seen as proofs that the last element of the smaller tree is present in the smaller and the larger tree, and a relationship between the two proofs.
For example, to prove that $T_2$ extends $T_1$ in  Fig.~\ref{fig:merkletree}, the proof of extension $\rho$ must include the following values:
\[
\rho = \makepe\bigl(R_3,\bigl((\myleft, h_{57}), (\myright, h_{12}), (\myleft, h_{4} ) \bigr) \bigr).
\]

The verification of the proof of extension consists of (i) verifying that the list $\bigl((\myleft, h_{57}), \allowbreak (\myright, h_{12}), \allowbreak (\myleft, h_{4} ) \bigr)$ is a proof of presence of $R_3$ in $T_2$; and (ii) that this  list filtered by only keeping the elements $(\myright, x)$ is a proof of presence of $R_3$ in $T_1$.


We refer the reader to the repository of the models~\cite{supplement} for the Horn clause definitions of the predicates $\represents$ and $\verifype$.

\subsection{Properties for the data structure}
Regardless of what data structure is used to model the ledger, the protocol expects that a number of
properties hold, which can be abstracted away from the particular data structure. 
In order to express these properties, we consider the same three predicates we presented above ($\verifypp$, $\verifype$ and $\represents$) and we axiomatize their intuitive semantics, i.e. $\verifypp$ validates proofs of presence, $\verifype$ validates proofs of extension and $\represents$ validates that a digest represents the contents of the data structure. 
We present the properties as first-order logical formulas as follows:

\medskip

\noindent\textbf{Proof for the empty list.} A data structure $h_0$
represents the situation where the list is empty:
\begin{align}
	&\forall h.\ \represents(\nil, h) \implies h = h_0.\tag{P1}\label{eq:p1}
\end{align}

\noindent\textbf{Correctness of the proof of presence.} $\mem(R,\ell)$ iff
its presence can be proved:
\begin{align}
	&\forall R, \ell, h, \pi.\notag\\
	&\ \ \represents(\ell, h) \wedge \mem(R,\ell) \implies \verifypp(\pi, R, h),\tag{P2}\\
	&\ \ \represents(\ell,h) \wedge \verifypp(\pi, R, h) \implies \mem(R,\ell).\tag{P3}
\end{align}

\noindent\textbf{Correctness of the proof of extension.} A data
structure for a list $\ell$ extends the data structure for
any of its suffixes:
\begin{align}
	\forall R, \ell, h, \rho. \
	&\represents(\cons(R, \ell), h) \implies\notag\\
	&\quad\exists h'. \represents(\ell, h') \wedge \verifype(\rho, h', h).\tag{P4}
\end{align}

\noindent\textbf{Transitivity of the proof of extension.} If $h_3$ extends $h_2$, which in turn extends $h_1$, then $h_3$ extends $h_1$:
\begin{align}
	&\forall \rho_1, \rho_2, h_1, h_2, h_3.\notag\\
	&\quad\verifype(\rho_1, h_1, h_2)\wedge \verifype(\rho_2, h_2, h_3) \implies\notag\\
	&\quad\quad \exists \rho_3. \verifype(\rho_3, h_1, h_3).\tag{P5}
\end{align}

\noindent\textbf{Compatibility of the proofs of extension and
  presence.} If an element is present in a list, then it remains present after
further elements have been added to it:
\begin{align}
	&\forall R, \pi_1, \rho, h_1, h_2.\notag\\
	&\quad\verifypp(\pi_1, R, h_1)\wedge \verifype(\rho, h_1, h_2) \implies\notag\\
	&\quad\quad \exists \pi_2. \verifypp(\pi_2, R, h_2).\tag{P6}\label{eq:p6}
\end{align}

\noindent\textbf{Consistency of digest representation.} If two lists are represented by the same digest, then they are equal:
\begin{align}
	&\forall \ell_1, \ell_2, h.\ \represents(\ell_1, h)\wedge \represents(\ell_2, h) \implies \ell_1 = \ell_2.\tag{P7}\label{eq:p7}
\end{align}

As described above, \eqref{eq:p1}-\eqref{eq:p7} can be regarded as an ``interface'' of security properties, which is instantiated as ProVerif queries
when we prove them for the data structure in Step~\ref{step:method1} of the methodology, and it is instantiated as ProVerif axioms with blocking predicates when they are
used to prove security in the protocol, in Step~\ref{step:method2}.


\subsection{Modelling the transparent-decryption protocol and its security properties}
\label{sec:modelling protocol}
\begin{figure}[!t]
\centering
\begin{lstlisting}
process
  ! 
  new $dk$; // Decryption key
  new $sk$; // Signing key
  out(c, $(\enckey(dk),\verkey(sk))$); // Output public keys
  new cell;    // Memory cell of the trustee
  new monitor; // Priv. chan. to store generated ciphertexts
  ( out(cell, ($0, h_0$)) // Initialize trustee memory cell
  	| ( ! // Trustee decrypt protocol
      in(cell, $(i, h)$);  	  
      in(c, $(R, \pi, \rho, h')$);
      if $\verifype(\rho, h, h')$ then
      if $\verifypp(\pi, R, h')$ then
      event Decrypted(cell,$i, R$);
      out(c, $\decrypt(dk, R)$);
      out(cell, $(i+1, h')$)
  	) | ( ! // Trustee sign hash protocol
      in(cell, $(i, h)$);
      in(c, $v$);
      event Signature(cell,$i, h, \sign(sk, (v, h))$);
      out(c, $\sign(sk, (v, h))$);
      out(cell, $(i, h)$)
    ) | ( ! // Subject ciphertext generation
      new $s$;
      event Secret($s$);   
      out(c, $\encrypt$($ek$, $s$));
      out(monitor, ($s$, $\encrypt$($ek$, $s$)))
    ) | ( ! // Monitor ciphertext
      in(monitor, ($s$, $R$));
      in(c, =$s$);
      new $v$;
      event Name($v$,$s$,$\enckey(dk)$,$\verkey(sk)$);
      out(c, $v$);
      in(c, $\sigma$);
      let $(=v, h) = \checksign(vk, \sigma)$ in
      event AfterSeeingSecret($R, h$)
    ))
\end{lstlisting}
\caption{Transparent decryption model (sketch)}
\label{fig:model}
\end{figure}
In this section, we provide the most relevant details of the model of the protocol for transparent decryption. A sketch of the model is depicted in Fig.~\ref{fig:model}. We omit the definitions of the required equational theories for hashing, public key encryption, signatures, types and events.
The model essentially consists of four sub-processes running in parallel.

The main process starts by creating private and public keys and initializing the public channel and two private channels. The requirement of private channels is because, whereas private constant values such as private keys can simply be defined within the scope of the trusted process as new values, the standard way to store mutable values in applied-pi calculus is through private channels. Thus, the channels used by our model are:
\begin{itemize}
\item Channel \lstinline|c|: public channel used for communication across the parties.
\item Channel \lstinline|cell|: private channel used as a memory cell for the trustee, in order to store the last seen hash value.
\item Channel \lstinline|monitor|: private channel used to pass secret values and ciphertexts from the subject party to a ``monitor process'' that will be discussed below.
\end{itemize}

Below we relate how the four sub-processes map to the different parties presented in Sec.~\ref{sec:transparent_decryption}. However, we remark that, as our trust assumptions are very weak, we do not require to model all the parties so that ProVerif can reason about the security of the protocol:

\begin{itemize}
    \item Subjects are represented by a sub-process that creates new secrets $s$ and output ciphertexts generated using the trustee's public encryption key $ek$. An event \lstinline|Secret($s$)| is required to be used in the definition of the security properties.
    \item Trustees run the two protocols from Fig.~\ref{fig:trustee}, and thus are modelled using two sub-processes. We do not consider threshold decryption schemes for this proof of concept. Note that the first and last step in each of these protocols consists in reading and writing to the trustee memory cell, respectively. Again, the events \lstinline|Decrypted(cell,$i, R$)| and \lstinline|Signature(cell,$i, h, \sigma$)| are used to define the security properties. The former is executed when the trustee accepts to decrypt $R$, and the latter is used when the trustee casts a signature of its last seen hash.
    \item The decryption requester and the log $L$ are untrusted entities, as discussed in Sec.~\ref{sec:transparent_decryption}. Hence, there is no need to model them, as the Dolev-Yao adversary will emulate their behaviour.
    \item Finally, a monitor process is required to ensure that the protocol works as intended and satisfies the claimed security properties. Its main task consists in interacting with the trustee through its second protocol (right column on Fig.~\ref{fig:trustee}) by generating a random nonce $v$ and obtaining the signature of the trustee's last seen hash. Once this happens, the event \lstinline|AfterSeeingSecret($R, h$)| is declared, stating that the monitor has a proof that the trustee has stored the hash value $h$ after the subject has created the ciphertext. 
\end{itemize}

The sub-processes are accordingly replicated to model arbitrary executions of the protocols, and arbitrary number of independent trustees that might monitor different logs.

%
%
%

The formalization of the main security property of the accountable decryption protocol, presented in Sec.~\ref{sec:transparent_decryption}, is
\begin{align}
  \forall R, h.\ &\mathrm{AfterSeeingSecret}(R, h) \nonumber\implies \\
                 &\quad\exists \pi.\ \verifypp(\pi, R, h). \label{eq:mainprop}
\end{align} 

Indeed, the location of \lstinline|AfterSeeingSecret($R, h$)| in the monitor process captures
the fact that a certain secret $s$ associated to ciphertext (decryption request) $R$ has been observed in the public channel, and the last observed hash value by the trustee is $h$. Therefore, any occurrence of this event means that there must exist a proof of presence $\pi$ stating that $R$ is in the data structure represented by $h$.

In order to prove this query, we consider some additional lemmas such as the following one:
\begin{align*}
  &\forall v, s, i, h,cell,sk,dk.\\
  &\ \mathrm{Signature}(cell,i, h, \sign(sk, (v, h))) \wedge\\ 
  &\ \ \mathrm{Name}(v, s,\enckey(dk),\verkey(sk)) \implies \\
  &\ \ \ \exists j.\ \mathrm{Decrypted}(cell,j, \encrypt(\enckey(dk), s)) \wedge j < i.
\end{align*}
Notice that the memory cell of the trustee stores two pieces of information: the number of times it decrypted a ciphertext, and the latest hash value it received. Hence this lemma states that when the trustee signed the name $v$ generated by the monitor (which monitors the trustee with encryption key $\enckey(dk)$ and verification key $\verkey(sk)$) after receiving a secret $s$, then the trustee must have decrypted it \emph{strictly before} (i.e. $j < i$).

This lemma allows us to help ProVerif by linking the content of the memory cell (i.e. the number of times it decrypted a encryption) with the order of events that were emitted in the trace. Such links are usually abstracted away by ProVerif during the saturation procedure hence the need for us to provide it within a lemma. Note that we also considered some additional lemmas and axioms, specific to the management of memory cells, in the vein of~\cite{cheval18}.

Our ProVerif models only take couple of seconds to execute on a standard laptop.

\input{proverif}

\section{Conclusion}
Transparency in security protocols plays a fundamental role in minimizing
the trust conditions for the parties involved. A clear example is that of certificate transparency, where strong security and trust assumptions on certificate authorities can be relaxed by requiring that they publish certificate issuances on a public ledger.
In many cases, transparency might well constitute a required building block expected by users of a certain service, especially when it involves incursions into their privacy for a variety of reasons, a situation which is handled by transparent decryption protocols.

%

For these reasons, it is important to properly verify the core
security property of transparency protocols, e.g. that only legitimate certificates
are produced (certificate transparency) or that decryptions
only take place if the requests for them are entered in a public
ledger, i.e. visible to users (transparent decryption). 

Because of the complex data structures that transparency protocols rely on, verifying
their properties has led to designing a proof-decomposition
methodology for ProVerif, as well as adding new features to the way
lemmas and axioms are handled. We expect that our methodology and
ProVerif enhancements can be applied to other kind of protocols
involving tree-based data structures (binary trees, radix trees) and
perhaps also other kinds of data structures (e.g. Bloom filters).

\section*{Acknowledgements}
This work received funding from EPSRC projects \emph{CAP-TEE: Capability Architectures for Trusted Execution};  \emph{SIPP: Secure IoT Processor Platform with Remote Attestation}, and \emph{User-controlled hardware security anchors: evaluation and designs}. It also received funding from the France 2030 program managed by the French National Research Agency under grant agreement No. ANR-22-PECY-0006.

\bibliographystyle{IEEEtran}
\bibliography{biblio}

\ifthenelse{\boolean{longversion}}{
\cleardoublepage
\appendices
\input{app-proverif}
}{}

\end{document}

%% file: lstsettings-proverif.tex
\definecolor{RoyalBlue}{rgb}{0.0, 0.14, 0.4}
\definecolor{Emerald}{rgb}{0.31, 0.78, 0.47} 
\definecolor{RedViolet}{rgb}{0.78, 0.08, 0.52}
\definecolor{SandyBrown}{rgb}{0.96, 0.64, 0.38}
\definecolor{Gray}{rgb}{0.5, 0.5, 0.5}
\definecolor{CadmiumGreen}{rgb}{0.0, 0.42, 0.24}

\colorlet{sqBlue}{RoyalBlue}
\colorlet{sqGreen}{Emerald}
\colorlet{sqRed}{red}
\colorlet{sqViolet}{RedViolet}
\colorlet{sqBrown}{SandyBrown}
\colorlet{sqGray}{Gray}
\lstdefinelanguage{sapic}{
  morekeywords=[1]{
	out, in, if, then, else, event, insert, delete, lookup, as, in, lock, unlock, let,
        new
      },
  morekeywords		= [2]{bitstring,pkey,skey,time,channel,t_key,t_input,t_output},
  morekeywords		= [3]{snd,fst,dec,enc,pk,h},
  morekeywords		= [4]{free, type, fun, const, reduc, equation,
    end, query, lemma, theory, begin, builtins, functions, equations,
    +, lemma, process, exists-trace, compfun, axiom, pred},
  morekeywords		= [5]{All, Ex, forall, otherwise, private,
    block, evalGround},
  morekeywords		= [6]{attacker, mess,K,KU,KD,A},
  sensitive=true,
  morestring=[b]',
  morestring=[s]{`}{'},
  morecomment		= [n][\itshape]{(*}{*)},
  morecomment		= [n][\bfseries]{(**}{*)},
    comment=[l]{//},
  literate=
        {||}{{$\mid$}}1
        {==>}{{$\implies$}}2
        {<--}{{$\leftarrow$}}2
        {->}{{$\rightarrow$}}2
        {\&}{{\textsf{\&}}}1
	{:=}{{$\defeq$}}2
        {'c}{{\textsl{c}}}3
}

\lstdefinestyle{sapic}{
  language={sapic},
  basicstyle		= \small\ttfamily,
  keywordstyle		= [2]{\mdseries\color{sqGreen}},
  keywordstyle		= [4]{\bfseries\color{sqViolet}},
  keywordstyle		= [3]{\mdseries},
  keywordstyle		= [1]{\bfseries\color{sqBlue}},
  keywordstyle		= [5]{\bfseries\color{sqBrown}},
  keywordstyle		= [6]{\bfseries},
  mathescape		= false,
  columns		= fullflexible,
  keepspaces		= true,
}



%% file: macro_commands.tex

\newcommand{\patt}{\mathsf{att}}										
\newcommand{\pevent}{\mathsf{event}}									
\newcommand{\pmsg}{\mathsf{mess}}									
\newcommand{\fatt}[1]{\patt(#1)}										
\newcommand{\fevent}[1]{\pevent(#1)}									
\newcommand{\fmsg}[2]{\pmsg(#1,#2)}									
\newcommand{\psevent}{\mathsf{s}\text{-}\mathsf{event}}				
\newcommand{\pmevent}{\mathsf{m}\text{-}\mathsf{event}}				
\newcommand{\sevent}[1]{\psevent(#1)}								
\newcommand{\mevent}[1]{\pmevent(#1)}								
\newcommand{\blocking}[1]{\mathsf{b}\text{-}#1}					

\newcommand{\qconcl}{\psi}						

\newcommand{\formula}{\phi}						

\newcommand{\tevalp}{\Downarrow'}


\newcommand{\V}{\mathcal{V}}            			
\newcommand{\N}{\mathcal{N}}            			
\newcommand{\Fc}{\mathcal{F}_c}         			
\newcommand{\Fd}{\mathcal{F}_d}         			
\newcommand{\Fe}{\mathcal{F}_e}         			
\newcommand{\Fp}{\mathcal{F}_p}         			
\newcommand{\Fap}{\mathcal{F}_{ap}}           
\newcommand{\Fbp}{\mathcal{F}_{bp}}         	
\newcommand{\Fdata}{\mathcal{F}_{data}}			  
\newcommand{\fail}{\mathsf{fail}}       			


\newcommand{\out}{\mathsf{out}}         			
\newcommand{\inn}{\mathsf{in}}          			
\newcommand{\new}{\mathsf{new}\ }       			
\newcommand{\llet}{\mathsf{let}\ }      			
\newcommand{\inelse}{\ \mathsf{in}\ }   			
\newcommand{\iif}{\mathsf{if}\ }        			
\newcommand{\then}{\ \mathsf{then}\ }
\newcommand{\eelse}{\ \mathsf{else}\ }
\newcommand{\event}{\mathsf{event}}     			
\newcommand{\suchthat}{\ \mathsf{suchthat}\ }	


\newcommand{\Cset}{\mathbb{C}}					

\newcommand{\mgu}[1]{\mathrm{mgu}(#1)}			

\newcommand{\select}{\mathsf{sel}}				
\newcommand{\unselectable}{\mathbb{F}_{usel}}	
\newcommand{\Deriv}{\mathcal{D}}					
\newcommand{\saturated}[1]{\mathsf{sat}(#1)}		

\newcommand{\Cinit}[1]{\Cset_\mathit{init}(#1)}			
\newcommand{\Cevents}[1]{\Cset_\mathit{ev}(#1)}			
\newcommand{\Cuser}{\Cset_\mathit{user}}
\newcommand{\CblockPred}{\Cset_{bp}}
\newcommand{\CblockU}{\Cset_\mathit{b-user}}
\newcommand{\CblockT}[1]{\Cset_{\mathsf{b}}(#1)}
\newcommand{\Cstd}{\Cset_\mathit{std}}
\newcommand{\Csat}{\Cset_\mathit{sat}}

\newcommand{\satisfy}{\vdash}					

\newcommand{\ToSEv}[1]{\lceil #1 \rceil^{sure}}	
\newcommand{\ToB}[1]{\lceil #1 \rceil^{b}}		
\newcommand{\measure}[1]{\mathsf{measure}(#1)}	

\newcommand{\Lemma}{\mathcal{L}}					

\newcommand{\dom}[1]{\mathrm{dom}(#1)}			
\newcommand{\vars}[1]{\mathrm{vars}(#1)}			


\newcommand{\GenClause}[3]{[\![#1]\!]#2#3}		
\newcommand{\Ch}{\mathcal{N}_{c}}
\newcommand{\Occ}{\mathcal{O}}
\newcommand{\Hfacts}{\mathcal{H}}
\newcommand{\occ}{o}
\newcommand{\I}{\mathcal{I}}
\newcommand{\D}{\mathcal{D}}
\newcommand{\subsume}{\sqsupseteq}
\newcommand{\subsumeselect}{\sqsupseteq_s}
\newcommand{\allowedpreds}{\mathcal{S}_p}
\newcommand{\satisfyD}{\vdash}					
\newcommand{\pred}{\mathrm{pred}}
\newcommand{\Pro}{\mathcal{P}}			
\newcommand{\Att}{\mathcal{A}}
\newcommand{\E}{\mathcal{E}}
\newcommand{\with}{\;||\;}
\newcommand{\defrw}{\mathsf{def}}		

\newcommand{\satisfyIO}{\vdash_\mathit{IO}}

\newcommand{\minsize}[2]{\mathrm{min}_{#1}(#2)}

\newcommand{\Cprotocol}{\mathbb{C}_{\mathcal{P}}}
\newcommand{\Cattacker}{\mathbb{C}_{\mathcal{A}}}

\def\C{\mathcal{C}}
\newcommand{\tracestep}[1]{\mathsf{trace}(#1)}
\newcommand{\tracestepIO}[1]{\mathsf{trace}_{IO}(#1)}
\newcommand{\istep}{\rightarrow_i}
\newcommand{\sem}[1]{\mathsf{sem}(#1)}
\newcommand{\names}{\mathrm{names}}		

\newcommand{\mem}{\mathsf{mem}}

\newcommand{\satisfyI}{\vdash_i}
\newcommand{\query}{\varrho}

\newcommand{\tuplesteps}{\tilde{\tau}}
\newcommand{\orderind}{<_{ind}}
\newcommand{\orderindeq}{\leq_{ind}}
\newcommand{\indlemma}[1]{{#1}^{ind}}

\newcommand{\IndHyp}[4]{\mathcal{IH}_{#1}(#2,#3,#4)}
\newcommand{\HypLemmas}[5]{\mathcal{H}yp_{#4,#5}(#1,#2,#3)}

\newcommand{\saturate}[4]{\mathsf{saturate}^{#2}_{#3,#4}(#1)}
\newcommand{\saturateS}[5]{\mathsf{saturateS}^{#3}_{#4,#5}(#1,#2)}

\newcommand{\Rlem}[1]{\text{Lem}$(#1)$\xspace}
\newcommand{\RlemOrd}[1]{\text{Lem}$_o(#1)$\xspace}
\newcommand{\Rind}[1]{\text{Ind}$(#1)$\xspace}
\newcommand{\RindOrd}[1]{\text{Ind}$_o(#1)$\xspace}

\newcommand{\less}{{<}}
\newcommand{\lesseq}{{\leq}}

\newcommand{\multiset}[1]{\{\!\!\{ #1 \}\!\!\}}		
\newcommand{\indexPred}[1]{\mathsf{idx}(#1)}

\newcommand{\strict}[1]{#1^{<}}

\newcommand{\RresOrd}[1]{\text{Res}$_o(#1)$\xspace}

\newcommand{\RNil}{\textsc{Nil}\xspace}
\newcommand{\RPar}{\textsc{Par}\xspace}
\newcommand{\RRepl}{\textsc{Repl}\xspace}
\newcommand{\RRestr}{\textsc{Restr}\xspace}
\newcommand{\RIO}{\textsc{I/O}\xspace}
\newcommand{\RMsg}{\textsc{Msg}\xspace}
\newcommand{\RLetin}{\textsc{Let1}\xspace}
\newcommand{\RLetelse}{\textsc{Let2}\xspace}
\newcommand{\ROut}{\textsc{Out}\xspace}
\newcommand{\RIn}{\textsc{In}\xspace}
\newcommand{\RApp}{\textsc{App}\xspace}
\newcommand{\RNew}{\textsc{New}\xspace}
\newcommand{\REvent}{\textsc{Event}\xspace}
\newcommand{\RPhase}{\textsc{Phase}\xspace}
\newcommand{\RInsert}{\textsc{Insert}\xspace}
\newcommand{\RGetin}{\textsc{Get1}\xspace}
\newcommand{\RGetelse}{\textsc{Get2}\xspace}
\newcommand{\RPredin}{\textsc{Pred1}\xspace}
\newcommand{\RPredelse}{\textsc{Pred2}\xspace}

%% file: abstract.tex
\begin{abstract}

Transparency protocols are protocols whose actions can be publicly monitored by observers (such observers may include regulators, rights advocacy groups, or the general public). The observed actions are typically usages of private keys such as decryptions, and signings. Examples of transparency protocols include certificate transparency, cryptocurrency, transparent decryption, and electronic voting. These protocols usually pose a challenge for automatic verification, because they involve sophisticated data types that have strong properties, such as Merkle trees, that allow compact proofs of data presence and tree extension.

We address this challenge by introducing new features in ProVerif, and a methodology for using them. With our methodology, it is possible to describe the data type quite abstractly, using ProVerif axioms, and prove the correctness of the protocol using those axioms as assumptions. Then, in separate steps, one can define one or more concrete implementations of the data type, and again use ProVerif to show that the implementations satisfy the assumptions that were coded as axioms. This helps make compositional proofs, splitting the proof burden into several manageable pieces. 
We illustrate the methodology and features by providing the first formal verification of the transparent decryption and certificate transparency protocols with a precise modelling of the Merkle tree data structure.
\end{abstract}

%% file: introduction.tex
\section{Introduction}
Many security protocols assume the availability of trusted third parties, such as cloud computing operators. These are entities which are inherently trusted by definition, and whose corruption or misbehavior might be undetected and have catastrophic consequences from a security standpoint. With the advent of web3 and distributed ledger technologies, there is an ever growing interest in the concepts of transparency and accountability of misbehaviour in security protocols: whereas we cannot prevent a malicious action from happening, we can track actions and take the appropriate corrective measures on the entities that have broken the presumed trust assumptions.

The focus of this paper is on transparency. In general terms, we say that a protocol is a \emph{transparency protocol} if (some of) its actions can be publicly monitored by a collection of observers, which can, at a later stage, provide evidence that certain actions occurred. By technical measures to ensure that the details of such actions are made available to relevant parties (whether the general public, regulators, rights advocacy groups or individuals), the protocol reduces the amount of trust required of trusted third parties, and deters malicious actions.

To motivate our work, we consider two prominent use cases of transparency protocols. On one hand, we consider the Google initiative on \emph{certificate transparency} \cite{CT} for monitoring and autiting digital certificate issuance which has become an IETF standard~\cite{rfc9162}. It enables the detection of illegitimate certificates that have been produced either erroneously or deliberately. On the other hand, we also consider \emph{transparent decryption}~\cite{Ryan:17}, a protocol that ensures visibility of decryption requests, and has applications in a variety of areas such as surveillance, data sharing and location-based services, among many others.



Appropriate transparency conditions for a protocol can be defined using an
append-only ledger, or a blockchain. We assume that some party is
willing to maintain the ledger, and we demonstrate how this can be
organised in such a way that the maintainer can produce data
demonstrating that it is, indeed, maintaining the log correctly. The ledger
contents are intended to be
accessible by any party to whom transparency is being
offered. Therefore, any such party can determine that the ledger is
running correctly (or that it is not running correctly, or that their
access to it has been denied).

\paragraph{Automatic verification}
Our paper concerns how to verify systems for transparency protocols.
That is, how to prove that the system really does guarantee that 
the actions of monitored parties are made available to observers.

Several tools have been proposed for automated analysis of security
protocols. Some of these tools impose restrictions on the protocols in
order to achieve termination of the analysis.  For example, the tools
may assume a bounded number of sessions, like Avispa
\cite{armando2005avispa}, DeepSec \cite{cheval2018deepsec}, or Akiss
\cite{chadha2016automated}. These tools are efficient at finding
attacks on small protocols but quickly face state explosion for
complex protocols. Hence for large and complex protocols, tools like
Tamarin \cite{meier2013tamarin} and ProVerif \cite{BlanchetFnTPS16} are often preferred. They
both offer a flexible framework to model a protocol and its
primitives, as well as their security properties. One key feature of
Tamarin is that it offers an interactive mode when the tool fails to
prove a protocol, while ProVerif typically offers more automation.

We work within the framework of ProVerif. It supports cryptographic primitives including
symmetric and asymmetric encryption; digital signatures; hash
functions; bit-commitment; and signature proofs of knowledge. The tool
is capable of evaluating secrecy properties, authentication
properties, and indistinguishability properties. In ProVerif, protocol
analysis is considered with respect to an unbounded number of sessions
and an unbounded message space. The tool is capable of attack
reconstruction: when a property cannot be proved, an execution trace
which falsifies the desired property can often be constructed.

 The transparency ledger may be implemented as a Merkle tree, which means that the ledger
maintainer can produce proofs that a data item is in the ledger
(\emph{proof of presence}), and that the ledger is only being appended
to (\emph{proof of extension}).  To model this in ProVerif, we use
user-defined predicates whose semantics is defined with Horn clauses.
ProVerif can work with arbitrary Horn clauses, but adding them often leads
to non termination of the ProVerif resolution strategy. Recent work~\cite{blanchet:hal-03366962} has extended ProVerif with notions of
lemmas and axioms, in an effort to address this non-termination issue.
However, lemmas and axioms have strong limitations with user-defined, attacker or message
predicates. We address these limitations in this paper.

\paragraph{Our contributions}
The paper develops new ProVerif capabilities and a methodology for using them.
Our contributions are as follows:
\begin{itemize}
\item We introduce new capabilities in ProVerif; more precisely, we define semantics and algorithms that allow ProVerif to work with lemmas and axioms that involve user-defined predicates.
\item We prove soundness of the algorithms, and implement them in a new version of ProVerif.
\item We introduce a new methodology for ProVerif, in which the proof of a protocol can be given based on assumptions about the behaviour of a data type (we code these assumptions as an \emph{interface}); then, in  separate steps we formally prove the assumptions hold for one or more concrete realisations of the data type. This methodology would syntactically not be possible without our extension of ProVerif.
\item We model two transparency protocols (transparent decryption and certificate transparency), and successfully instantiate the proposed methodology.
\end{itemize}
The paper is supported by our new version of ProVerif and the ProVerif
scripts for transparent decryption, 
which can be found in \cite{supplement}. 
\ifthenelse{\boolean{longversion}}{
Our code has been reviewed by the ProVerif owners, and it will be incorporated in the next ProVerif release. Detailed definitions and proofs are provided in Appendix.
}{
Our code has been reviewed by the ProVerif owners, and it will be incorporated in the next ProVerif release. Detailed definitions and proofs are provided in~\cite{tech}.
}

%% file: proverif.tex
\section{Extending ProVerif to support arbitrary predicates in lemmas and axioms}

\newcommand{\Decrypted}{\mathrm{Decrypted}}
\newcommand{\Secret}{\mathrm{Secret}}
\newcommand{\ek}{{ek}}
\newcommand{\s}{{s}}
\newcommand{\celldevice}{{d}}
\newcommand{\ProVerif}{ProVerif\xspace}
\newcommand{\minstep}[2]{\mathsf{min}_{#1}(#2)}
\newcommand{\minstepU}[3]{\mathsf{min}_{#1,#2}(#3)}
\newcommand{\invDeriv}[2]{\mathsf{Inv}_{#1}(#2)}
\newcommand{\invDerivU}[3]{\mathsf{Inv}_{#1,#2}(#3)}
\newcommand{\verifypeSimple}{\mathsf{v}_\mathsf{pe}}

To complete the methodology depicted in Fig.~\ref{fig:methodology} on our running example, we need to 
\begin{inparaenum}[(1)]
\item prove the properties~\ref{eq:p1} to~\ref{eq:p7}, and
\item prove the main protocol while expressing the properties~\ref{eq:p1} to~\ref{eq:p7} as axioms.
\end{inparaenum}
All the properties in our interface are in fact correspondence properties that are within the scope of ProVerif. For example, the query corresponding to Property~\ref{eq:p6} would be expressed as follows:
\begin{lstlisting}
query pe1,pe2,pe3:proof_of_extension,
      d1,d2,d3:digest; 
verify_pe(pe1,d1,d2)&&verify_pe(pe2,d2,d3)
      ==> verify_pe(pe3,d1,d3) 
\end{lstlisting}

However, in its current version, ProVerif imposes a syntactic restriction on axioms and lemmas: the conclusion of a lemma can only contain events, equalities, disequalities and blocking user-defined predicates. Specifically, the native facts $\fatt{M}$ and $\fmsg{M}{N}$ as well as the clause-based user-defined predicates cannot be used in the conclusion of a lemma. 

This restriction of facts prevents us from achieving both steps of our methodology. The second step is unattainable as ProVerif will directly reject such a query if it is written as an axiom. For the first step, the query will be accepted by ProVerif but it will fail to prove it. Such a query requires a proof by induction, which internally corresponds to transforming a query into an \emph{inductive lemma} that has the same syntactic limitation as declared axioms and lemmas.
By extending ProVerif to allow any predicates in the conclusion of axioms and lemmas, we are able to complete our methodology.

In this section, we provide a high-level description of ProVerif's procedure and how we extended it.


\subsection{Description of ProVerif's procedure}

Horn clauses are used to describe the semantics of user-defined predicates as previously described but they are also the building blocks of \ProVerif's internal procedure to prove a secrecy property and more generally a correspondence property. Specifically, \ProVerif first translates the protocol given as input into a set $\Cset$ of Horn clauses. It then proceeds to \emph{saturate} this set $\Cset$ and the clauses that define the user-defined predicates, yielding a simpler set of clauses that derives the same facts. The procedure completes by verifying that the saturated clauses satisfy the security property. 



\paragraph{Translation into Horn clauses}
In addition to the user-defined predicates, \ProVerif considers natively four additional predicates over terms representing the interactions between the attacks and the processes: $\fatt{M}$, $\fmsg{M}{N}$, and two predicates for events $\sevent{ev}$ and $\fevent{ev}$. \emph{Sure-events}, i.e. $\sevent{ev}$, will only appear in hypotheses of Horn clauses, whereas events $\fevent{ev}$ will only appear in their conclusion.
This separation ensures that events are not \emph{resolved} during the saturation procedure and so their occurrence in a Horn clause is preserved through resolution rule.

Using these predicates, \ProVerif generates a set of clauses representing the capabilities of the attacker, which include, for example:
\begin{align}
\fatt{x} \wedge \fatt{y} &\rightarrow \fatt{\encrypt(x,y)} \label{clause-build}\\
\fatt{\encrypt(x,y)} \wedge \fatt{x} &\rightarrow \fatt{y} \label{clause-decrypt}\\
\fatt{x} \wedge \fmsg{x}{y} &\rightarrow \fatt{y}\label{clause-mess}\\
\fatt{x} \wedge \fatt{y} &\rightarrow \fmsg{x}{y}
\end{align}
The first two clauses model that the attacker can encrypt and decrypt provided that it knows the secret key. The last two clauses model that the attacker can read and write on a channel that it knows.

The formal description of processes into Horn clauses is out of scope of this paper (see~\cite{blanchet:hal-03366962} for more details), but we provide some intuition in the following example.

\begin{example}
The translation of the \lstinline|Ciphertext Generator| process would yield at least the clause 
\begin{equation}
\sevent{\Secret(\s)} \rightarrow \fatt{\encrypt(\ek, \s)}	\label{clause-generator}
\end{equation}
indicating that the attacker can obtain the encryption of the secret $\s$ by the key $\ek$. The Horn clause also indicates that the event $\Secret(\s)$ is triggered before the encryption is sent to the attacker.   

Similarly, the translation of process modelling the accountable decryption device will generate, in particular, the clause:
\begin{multline}\label{clause-accountable}
\fmsg{\celldevice}{(i,H_0)} \wedge \fatt{\encrypt(\ek, x)} \wedge \fatt{(pi,r,H_1)} \wedge\\
\verifype(r,H_0,H_1) \wedge \verifypp(pi,\encrypt(\ek, x),H_1) \wedge\\
\sevent{\Decrypted(i,\encrypt(\ek, x))} \rightarrow \fatt{x}
\end{multline}
This clause represents the informal statement: Assuming that the cell of the device contains the $i$-th digest received $H_0$, if the attacker can provide a proof of extension $r$ from $H_0$ to a new digest $H_1$, a proof of presence $pi$ of some ciphertext $\encrypt(\ek, x)$ in the digest $H_1$ then the attacker can obtain the plain text $x$. As in the previous clause, the event $\Decrypted(i,\encrypt(\ek, x))$ will be triggered before the attacker obtains $x$.
\end{example}


\paragraph{Saturation}
The core step in the saturation procedure consists of taking two existing Horn clauses and combining them into a new one, hopefully simpler. This process is called the \emph{resolution step}. For example, the hypothesis of clause~\ref{clause-accountable} contains in hypothesis the fact $F = \fatt{\encrypt(\ek, x)}$ representing that the attacker must know the ciphertext $\encrypt(\ek, x)$. In order to deduce all the possible ways the attacker may deduce this ciphertext, \ProVerif resolves this fact by combining it with Horn clauses whose conclusion can be unified with $F$. This is the case with clause~\ref{clause-generator}, which results in the following clause:
\begin{multline*}
\fmsg{\celldevice}{(i,H_0)} \wedge \sevent{\Secret(\s)} \wedge \fatt{(pi,r,H_1)} \wedge\\
\quad \verifype(r,H_0,H_1) \wedge \verifypp(pi,\encrypt(\ek, \s),H_1) \wedge\\
\quad \sevent{\Decrypted(i,\encrypt(\ek, \s))} \rightarrow \fatt{\s}.
\end{multline*}
Note that the hypothesis $\sevent{\Secret(\s)}$ of clause~\ref{clause-generator} has replaced $F$ in the new clause and the unification of the two facts resulted in $x$ being instantiated by $\s$. 

In our example, the fact $F$ in clause~\ref{clause-accountable} could also have been resolved with clause~\ref{clause-decrypt}, which would result in
a new rule like clause~\ref{clause-accountable} but with $F$ replaced by
the two facts $\fatt{\encrypt(x', \encrypt(\ek, x))}$ and $\fatt{x'}$.
This would directly lead to a loop during the saturation procedure as the former fact could once again be resolved by the same clause~\ref{clause-decrypt}. To avoid this problem, ProVerif uses a \emph{selection function} on Horn clauses that returns the set of facts from the hypothesis of a clause that can be resolved, i.e. $\select(H \rightarrow C) = S$ with $S \subseteq H$. In particular, the facts $\fatt{x}$ with $x$ a variable are never in $S$. 

The resolution process repeats until a fixed point is reached, i.e. until no resolution can produce a clause that is not redundant with an existing one. Once a fixed point is reached, ProVerif only keeps the set of \emph{saturated clauses}, which are the clauses $R$ from $\Cset$ such that $\select(R) = \emptyset$.
A schematic summary of \ProVerif's saturation procedure is given in Fig.~\ref{fig:saturation}.


\begin{figure}
\begin{tikzpicture}[
	state/.style={draw,rounded corners,align=center}
	]
	\node[state, text width = 1.2cm] (set)	{Current set $\Cset$};
	\node[state,above right=0.7cm and 0.4cm of set,text width=4cm] (resol) {Resolution of two clauses of $\Cset$ producing $R$};
	\node[state,right=4.4cm of set,text width=2.2cm] (simpl) {Simplification of $R$ into $\Cset_R$};
	\node[state,below right=0.7cm and 0.4cm of set,text width=4cm] (subsumption) {Filtering of $\Cset \cup \Cset_R$};
	
	\node[state,right=1.3 cm of set,text width=2.7cm] (sat) {Saturated set $\Csat = \{ R \in \Cset \mid \select(R) = \emptyset\}$};
	
	\draw[->] (set) edge node[auto] {Fix point} (sat);
	\draw[->] (set) edge node[auto,below] {reached} (sat);
	\draw[->] (set) to [in=180,out=90] (resol) {};
	\draw[->] (resol) to [in=90,out=0] (simpl) {};
	\draw[->] (simpl) to [in=0,out=270] (subsumption) {};
	\draw[->] (subsumption) to [in=270,out=180] (set) {};
\end{tikzpicture}
\caption{A schematic summary of ProVerif's saturation procedure}
\label{fig:saturation}
\end{figure}


\paragraph{Generalizing the events and sure-events separation} We previously mentioned that ProVerif considers two different predicates for events: the sure-event predicate $\psevent$ and the event predicate $\pevent$. During the translation of processes into Horn clauses, the former type of event only occurs in the hypotheses of the clauses whereas the latter only occurs in the conclusion of the clauses. It ensures that an application of the resolution rule never resolves a sure-event $\sevent{ev}$ from a clause. This is critical for the verification of correspondence queries. For example, to verify the query $\fevent{A} \Rightarrow \fevent{B}$, ProVerif will check that for all saturated clauses in $\Csat$ of the form $H \rightarrow \fevent{A}$, the sure-event $\sevent{B}$ occurs in $H$. In this case, ProVerif concludes that the query holds. 
Note that to ensure soundness, the selection function will never consider sure-events, i.e. $\sevent{ev} \not\in \select(R)$, which guarantees that the resolution rule does not try to resolve a fact that cannot be resolved by design.

Since sure-events are never selected by the selection function, they can be seen as the \emph{blocking} counterpart of the predicate $\event$.
To add clause-based, user-defined predicates in the conclusion of lemmas, as well as the predicates $\fatt{M}$ and $\fmsg{M}{N}$, we generalize this concept by associating to all predicates $p$ a blocking predicate denoted $\blocking{p}$. Therefore, we consider natively the predicates $\blocking{\patt}$ and $\blocking{\pmsg}$. Moreover, for all user-defined predicates $p \in \Fp$, we consider the blocking predicate $\blocking{p}$ and denote by $\Fbp$ their set. Finally, we rename the predicate $\psevent$ as $\blocking{\pevent}$. 

We also amend the selection function by additionally requiring that no blocking fact can be selected: For all $F$, $\blocking{F} \not\in \select(H \rightarrow C)$. Thus, the resolution rule will never attempt to resolve a blocking predicate.


\paragraph{Application of lemmas} ProVerif axioms and proved lemmas are applied during the simplification phase of the saturation procedure. Consider a lemma $\bigwedge_{i=1}^n F_i \Rightarrow \bigvee_{j=1}^m \qconcl_j$, where $\qconcl_j$ are conjunctions of facts. In the current ProVerif, each $\qconcl_j$ could only be composed of events, dis\-equalities and inequalities. The lemma would be applied on a clause $H \rightarrow C$ when there exists a substitution $\sigma$ such that $F_i\sigma$ is in $H$ for all $i \in \{1,\ldots,n\}$. The application of the lemma would then produce a set of $m$ clauses $\{ H \wedge \qconcl_j'\sigma \rightarrow C\}_{j=1}^m$, where each $\qconcl_j'$ is the disjunct $\qconcl_j$ with events replaced by their sure-event counterpart. 

Thanks to our extension, the disjunct $\qconcl_j$ may now contain any type of predicate. We update the simplification rule by first defining the transformation  $\ToB{\qconcl}$ built from $\qconcl$ where all facts $p(M_1,\ldots,M_n)$ in $\qconcl$ are replaced by $\blocking{p}(M_1,\ldots,M_n)$. The rule for applying lemmas is then defined as follows:
\begin{prooftree}
	\AxiomC{
	\begin{tabular}{@{}c@{}}
	$\Cset \cup \{ R=(H \rightarrow C)\}$\qquad
	$(\bigwedge_{i=1}^n F_i \Rightarrow \bigvee_{j=1}^m \qconcl_j) \in \Lemma$\\
	$\forall i,\blocking{F_i\sigma} \in H\text{ or }F_i\sigma \in H$
	\end{tabular}
	}
	\UnaryInfC{$\Cset \cup \{ H \wedge \ToB{\qconcl_j\sigma} \rightarrow C \}_{j=1}^m$}
\end{prooftree}

Note that the application condition requires that either $\blocking{F_i\sigma} \in H$ or $F_i\sigma \in H$. Indeed, the lemma can be applied on facts introduced by a previous application of another lemma or more commonly on an event, hence the condition $\blocking{F_i\sigma} \in H$. Since a lemma may have an attacker fact, a message fact or a fact using a user predicate in its premisse, we also need to match the non-blocking form of the fact, i.e. $F_i\sigma \in H$. Note that all facts added in the clauses, i.e. facts in $\ToB{\qconcl_j\sigma}$, are blocking.

\paragraph{Improving other simplification rules} The purpose of applying lemmas is to increase ProVerif precision or to help it terminate. As such, we can amend other simplification rules used by ProVerif to benefit from the blocking facts. For example, ProVerif employs the following simplification rule to remove tautologies:

\begin{prooftree}
	\AxiomC{$\Cset \cup \{ F \wedge H \rightarrow F\}$}
	\UnaryInfC{$\Cset$}
	\RightLabel{Taut}
\end{prooftree}

Now that we may have blocking predicates in the hypotheses of the clause, we can consider an additional tautology simplification rules defined as follows:

\begin{prooftree}
	\AxiomC{$\Cset \cup \{ \blocking{F} \wedge H \rightarrow F\}$}
	\AxiomC{$\pred(p) \not\in \Fp$}
	\BinaryInfC{$\Cset$}
\end{prooftree}

We also improve the simplification rule that removes the redundant facts from the hypotheses of a clause, which is critical to avoid termination issues:


\begin{prooftree}
	\AxiomC{
	\begin{tabular}{@{}c@{}}
	$\Cset \cup \{ H' \wedge H \wedge \phi \rightarrow C\}$\qquad
	$\ToB{H'\sigma} \subseteq \ToB{H}$\\
	$\phi \models \phi\sigma$\qquad
	$\dom{\sigma} \cap \vars{H,C} = \emptyset$
	\end{tabular}
	}
	\UnaryInfC{$\Cset \cup \{ H \wedge \phi\sigma \rightarrow C\}$}
\end{prooftree}
Intuitively, this rule states that to derive $C$, it suffices to know the derivations of $H$, as a derivation for $H'$ can be build from the derivations of $H'\sigma \subseteq H$. Note that by requiring $\ToB{H'\sigma} \subseteq \ToB{H}$, a clause $\fatt{M} \wedge \blocking{\patt}(M) \wedge G \rightarrow C$ can be either simplified into $\fatt{M} \wedge G \rightarrow C$ or into $\blocking{\patt}(M) \wedge G \rightarrow C$. Both variants are correct and their application may be parametrized depending on one's needs. The former could be used to discard the application of a lemma that is redundant with the hypothesis of the clause whereas the latter could help terminate, as the fact $\blocking{\patt}(M)$ will not be resolved.


\paragraph{Verification of the query} Once the saturation process ends, ProVerif still needs to verify the queries. Intuitively, on a query $F_1 \wedge \dots \wedge F_n \Rightarrow \qconcl$ with $F_i = p_i(t_1^i,\ldots,t^i_{m_i})$, ProVerif starts by generating a clause $F_1 \wedge \dots \wedge F_n \rightarrow C$ with $C = \mathsf{q}(t^1_1,\ldots,t^n_{m_n})$ and $\mathsf{q}$ a special predicate only used in the verification of query. Typically, $C$ represents the conjunction of facts $F_1,\ldots,F_n$. ProVerif then applies once again a saturation of $\Csat \cup \{ F_1 \wedge \dots \wedge F_n \rightarrow C \}$ and checks the validity of $\qconcl$ on the obtained saturated set. Note that this second saturation is slightly different from the first saturation in the sense that it is \emph{order preserving}, i.e. in a clause $F \wedge H \rightarrow \mathsf{q}(t^1_1,\ldots,t^n_{m_n})\sigma$, Proverif can indicate, for any $i$, whether the fact $F$ was generated when resolving $F_i$. This property is particularly important for proving queries by induction. The details of how the orders are preserved can be found in~\cite{blanchet:hal-03366962}.


\begin{example}
  Let us show how we are able to prove by induction the transitivity of proof of extension in the hash list data structure:
  \begin{lstlisting}
query pe1,pe2,pe3:proof_of_extension,
      d1,d2,d3:digest;
verify_pe(pe1,d1,d2)&&verify_pe(pe2,d2,d3)
     ==> verify_pe(pe3,d1,d3) [induction].
  \end{lstlisting}
  To ease the reading, let us denote by $\verifypeSimple$ the predicate $\verifype$. Recall that the proof of extension is defined by the following two clauses. 
  \begin{align}
  & \verifypeSimple(\makepe(\nil), d, d), \label{clause-VPE-empty}\\
  & \verifypeSimple(\makepe(\ell),d_1,d_2) \rightarrow \verifypeSimple(\makepe(\cons(x,\ell)),d_1,H(x,d_2)).\label{clause-VPE-step}
  \end{align}
 These clauses will be left unchanged through the saturation procedure, i.e. if $\Cset = \{ (\ref{clause-VPE-empty}),(\ref{clause-VPE-step})\}$ then $\saturated{\Cset} = \Cset$.

  To verify the query, ProVerif first considers the query clause $RQ = F_1 \wedge F_2 \rightarrow C$ where $F_1 = \verifypeSimple(pe_1,d_1,d_2)$, $F_2 = \verifypeSimple(pe_2,d_2,d_3)$ and $C = \mathsf{q}(pe_1,d_1,d_2,pe_2,d_2,d_3)$. ProVerif then applies the saturation procedure on $\Cset \cup \{ RQ \}$.
  
  We illustrate the application of the inductive lemma on one of the resolutions that ProVerif will apply during this second saturation procedure: the resolution of $\verifypeSimple(pe_2,d_2,d_3)$ from $RQ$ with the clause~(\ref{clause-VPE-step}) which yields the following clause
  \[
  RQ' = F_1\sigma \wedge \verifypeSimple(\makepe(\ell),d_2,d_3') \rightarrow C\sigma 
  \] 
  with $\sigma = \{ pe_2 \mapsto \makepe(\cons(x,\ell)), d_3 \mapsto H(x,d_3')\}$. 
  
  Let us denote $F_3 = \verifypeSimple(\makepe(\ell),d_2,d_3')$. Note that $C\sigma$ represents the conjunction $F_1\sigma \wedge F_2\sigma$ with:
  \[F_2\sigma = \verifypeSimple(\makepe(\cons(x,\ell)),d_2, H(x,d_3'))\]
  
  Since $F_3$ was obtained while resolving $F_2\sigma$, any instantiation $F_3\alpha$ would be \emph{satisfied strictly before} $F_2\sigma\alpha$ (we detail this notion in the next section). Thus, ProVerif can apply our inductive hypothesis on $F_1\sigma,F_3$ and so it will add the blocking fact $\blocking{\verifypeSimple}(pe',d_1,d_3')$ in the hypothesis of the clause, yielding:
  \[
  RQ_2 = F_1\sigma \wedge \verifypeSimple(\makepe(\ell),d_2,d_3') \wedge \blocking{\verifypeSimple}(pe',d_1,d_3') \rightarrow C\sigma
  \]
  On the clause $RQ_2$, ProVerif will be able prove the query, i.e. finding a derivation of $\verifypeSimple(pe_3,d_1,d_3)\sigma = \verifypeSimple(pe_3,d_1,H(x,d_3'))$ for some $pe_3$, by using the blocking fact $\blocking{\verifypeSimple}(pe',d_1,d_3')$ and the clause~(\ref{clause-VPE-step}).
  \end{example}


\subsection{Soundness}

The soundness of the saturation procedure comes in three steps. 

\paragraph{Derivation of satisfiable facts} Given a process $P$ and a trace $T$ of $P$, one can show that that for all satisfiable facts $F$ in $T$, i.e. $T \satisfy F$, there exists a derivation $\Deriv$ of $F$ in the set of initial clauses $\Cinit{P}$ translated from $P$ and the set of clauses $\CblockT{T} = \{ \ToB{F} \mid T \satisfy F \}$. The original soundness result~\cite[Theorem 1]{blanchet:hal-03366962} was only considering the sure-events, i.e. $\{ \sevent{ev} \mid T \satisfy \fevent{ev}\}$. 

Since blocking predicates are an artifice to prevent resolution within the saturation procedure, they should not affect their satisfiability in a trace $T$. Hence, we augment the satisfaction relation $\satisfy$ by requiring that for all blocking facts $\blocking{F}$, $T \satisfy \blocking{F}$ if and only if $T \satisfy F$.

In the current ProVerif, the derivation $\Deriv$ is shown to satisfy an invariant on the Horn clauses labeling its nodes that relates satisfaction of facts with the size of the trace. Consider a ground fact $F$ such that $T \satisfy F$. We define $\minstep{T}{F} = \min\{ | T'| \mid T' \text{ prefix of } T \wedge T' \satisfy F\}$ which represents the size of the smallest prefix of $T$ that satisfies $F$. The invariant intuitively indicates that in a derivation, all instantiated facts are satisfied by $T$ and if a clause $H \rightarrow C$ is used to derive $C$ then the facts in $H$ must be satisfied strictly before $C$ in the trace $T$. 
We extend this invariant to take into account the user-defined predicates and blocking predicates 

However, when a fact $F$ corresponds to a user-defined predicates, e.g. $F = p(M_1,\ldots, M_n)$, we always have $\minstep{T}{F} = 0$ as the satisfiability of this predicate do not depend on the protocol but only on the value of the terms $M_1,\ldots, M_n$. Hence, $T \satisfy F$ if and only if $F$ is satisfied in the empty trace. 
This prevents us from \emph{ordering} the user-defined predicates within a clause. Hence, we strengthen the minimality function on user-defined predicates by considering the size of the derivation of $p(M_1,\ldots, M_n)$ in $\Cuser$ instead of its satisfiability in $T$. Formally, we consider $\minstepU{T}{\Cuser}{F}$ defined as:
\begin{itemize}
	\item $\minstepU{T}{\Cuser}{F} = \minstep{T}{F}$ when $\pred(F) \not\in \Fp \cup \Fbp$
	\item $\minstepU{T}{\Cuser}{F} = \min\{ |\Deriv| \mid \Deriv$ derives $F'$ from $\Cuser\}$ when $F = F'$ or $F = \blocking{F'}$ with $\pred(F') \in \Fp$.
\end{itemize}

We can show state the main invariant on derivations.

\begin{invariant}
  \label{inv:derivation2}
  Let $T$ be a trace. Let $\Cuser$ a set of clauses defining predicates in $\Fp$. Let $\Deriv$ a derivation. We say that the invariant $\invDerivU{T}{\Cuser}{\Deriv}$ holds when for all nodes of $\Deriv$ labeled by a Horn clause $R = H\rightarrow C$ and a substitution $\sigma$, $T \satisfy C$ and for all $F \in H$, $T\satisfy F$. Moreover, 
    \begin{itemize}
      \item if $R$ is the attacker rule (\ref{clause-mess}) then $\minstep{T}{\fatt{x}\sigma} < \minstep{T}{\fatt{y}\sigma}$ and $\minstep{T}{\fmsg {x}{y}\sigma} = \minstep{T}{\fatt{y}\sigma}$;
      \item otherwise for all $F \in H$, if ($\pred(C) \not\in \Fp$ and $\pred(F) \not\in \Fp \cup \Fbp$) or ($\pred(C) \in \Fp$ and $\pred(F) \in \Fp$) then $\minstepU{T}{\Cuser}{F\sigma} < \minstepU{T}{\Cuser}{C\sigma}$
    \end{itemize}
  \end{invariant}
 As previously mentioned, the invariant intuitively states that all facts in the derivation are satisfied in $T$. The first bullet point indicates that in the case of the attacker rule (\ref{clause-mess}), representing that the attacker may listen on a channel if it can deduce it, the two facts $\fmsg{x}{y}$ and $\fatt{y}$ are satisfied at the same moment on $T$. The second bullet point indicates that for all other rules, the hypohteses are satisfied strictly before the conclusion, in the sense of $\minstepU{T}{\Cuser}{\cdot}$ and if they belong to the same category, i.e. standard or user-defined predicates.

\paragraph{Preservation of the invariant during the saturation}
In the second step of the soundness proof, one can show that one round \emph{resolution-simplification-filtering} preserves the derivability of facts, i.e. if $\Cset'$ is the new set of clauses after a round of \emph{resolution-simplification-filtering} on $\Cset$, then $F$ derivable in $\Cset$ implies that $F$ is derivable in $\Cset'$, with both derivations satisfying Invariant~\ref{inv:derivation2}.

The soundness of the new tautology rule is directly given by Invariant~\ref{inv:derivation2}. Indeed, if $\pred(F) \not\in \Fp$ then $\pred(\blocking{F}) \not\in \Fp \cup \Fbp$ and so $\minstepU{T}{\Cuser}{\blocking{F}} < \minstepU{T}{\Cuser}{F}$ which contradicts minimality.

Note that the soundness of the new rule for redundant facts is direct from the fact that invariant~\ref{inv:derivation2} is stable by removal of hypotheses. 

Finally, we show that the application of lemmas is also sound. This is achieved thanks again to Invariant~\ref{inv:derivation2}. In particular, when a lemma $\bigwedge_{i=1}^n F_i \Rightarrow \bigvee_{j=1}^m \qconcl_j$ matches $H$, i.e. $F_i\sigma$ is in $H$ for all $i \in \{1,\ldots, n\}$, it implies that each $F_i\sigma$ are satisfied in $T$. Since the lemma holds on $T$, at least one of the $\qconcl_j\sigma$ also holds in $T$ and so all facts in $\ToB{\qconcl_j\sigma}$ are derivable from $\CblockT{T}$.


\paragraph{Application of inductive lemmas} In~\cite{blanchet:hal-03366962}, a query $F_1 \wedge \dots \wedge F_n \Rightarrow \qconcl$ may be proved by induction on the multiset $\multiset{\minstep{T}{F_1\sigma},\ldots,\minstep{T}{F_n\sigma}}$ when $T \satisfy F_i\sigma$ for all $i = 1,\ldots, n$. In practice, ProVerif transforms the query into an \emph{inductive lemma} which is typically a lemma implied by the query. This lemma will be applied as a normal lemma during the saturation procedure except on the attacker rule~(\ref{clause-mess}). 

The main soundness argument for applying inductive lemma during the saturation procedure is as follows. As we prove the query by induction on $\mathcal{M} = \multiset{\minstep{T}{F_1\sigma},\ldots,\minstep{T}{F_n\sigma}}$, we can assume that the query holds for any instance of the premise $F_1\sigma',\ldots, F_n\sigma'$ such that $\multiset{\minstep{T}{F_1\sigma'},\ldots,\minstep{T}{F_n\sigma'}} < \mathcal{M}$. Now consider a rule $H \rightarrow C$ used in a derivation of $F_1\sigma,\ldots,F_n\sigma$, we know from Invariant~\ref{inv:derivation2} that all facts $F$ in $H$ are satisfied strictly before $C$ which is itself satisfied strictly before at least one of the $F_i\sigma$, i.e. $\minstep{T}{F} < \minstep{T}{F_i\sigma}$. Therefore, if $F_1\sigma',\ldots,F_n\sigma' \in H$ then we deduce that $\multiset{\minstep{T}{F_1\sigma'},\ldots,\minstep{T}{F_n\sigma'}} < \mathcal{M}$ and so we know that the conclusion of the inductive lemma holds, which allows us to apply it.

With user-defined predicates, we cannot use $\minstep{T}{\cdot}$ because the satisfiability of a user-defined predicate does not depend on the trace but on the set of Horn clauses $\Cuser$ given as input. Instead, we base the induction on the minimal size of the derivation of $F_i$ in $\Cuser$, i.e. $\minstepU{T}{\Cuser}{F_i\sigma}$.
More specifically, if for all $i = 1,\ldots, n$, $\pred(F_i)$ is a user-defined predicate, we prove the query by induction on the multiset $\multiset{\minstepU{T}{\Cuser}{F_i\sigma}}_{i=1}^n$.

To prove a query by induction whose premises contain both user-defined predicates and standard predicates, i.e. $\patt$ and $\pmsg$, we combine the two inductive measures into a single one: without lost of generality, consider the query  $F_1 \wedge \dots \wedge F_n \Rightarrow \qconcl$ where $\pred(F_i),\dots,\pred(F_{k-1})$ are standard predicates and 
$\pred(F_k),\dots,\pred(F_{n})$ are user-defined predicates. We prove the query for all traces $T$ and all substitutions $\sigma$ with $T \satisfy F_i\sigma$, for all $i=1,\ldots, n$, by induction using the lexicographic order on $(\multiset{\minstepU{T}{\Cuser}{F_i\sigma}}_{i=1}^{k-1}, \multiset{\minstepU{T}{\Cuser}{F_i\sigma}}_{i=1}^n)$.

\paragraph{Restriction to $\Csat$} The third and final step consists of showing that when the fix point is reached, restricting the $\Cset$ to $\Csat$ also preserves derivability. This step remains unchanged from~\cite[Theorem 2]{blanchet:hal-03366962}.

The general soundness of the saturation procedure is given by the following property.

 
\begin{theorem}
  \label{th:saturation soundness_simplified}
  Let $P$ a process and $T$ a trace of $P$.  Let $\Lemma$ a set of lemmas that hold on $P$. Let $\Lemma_i$ be a set of inductive lemmas. We denote by $\Cinit{P}$ the set of initial clauses translated from $P$. Let $F$ be a fact such that $T \satisfy F$. Let $\mathcal{M} = (\emptyset,\minstepU{T}{\Cuser}{F})$ when $\pred(F) \in \Fp$ and $\mathcal{M} = (\minstepU{T}{\Cuser}{F},\emptyset)$ otherwise. 
  
  If the lemmas in $\Lemma_i$ hold up to $\mathcal{M}$ excluded then there exists a derivation $\Deriv$ of $F$ from $\saturated{\Cinit{P}} \cup \CblockT{T}$ such that $\invDerivU{T}{\Cuser}{\Deriv}$.
\end{theorem}
  
This theorem is a simplified version of \cite[Theorem 1 and 2, simplified]{blanchet:hal-03366962} adapted to user-defined predicates and generalised lemmas. In particular, the key difference between these two soundness results is that we consider in the set $\CblockT{T}$ the blocking counter part of all satisfiable facts by the trace $T$. Note that $\CblockT{T}$ also includes the blocking counter part of all true instances of predicates $p(M_1,\ldots,M_n)$ where $p$ is a user-defined predicate.
  

\subsection{Applications}

As we previously showed, generalizing lemmas, axioms and inductive proofs is paramount for our new methodology. Thanks to the inductive proofs, we are able to show with ProVerif all the properties~\ref{eq:p1} to~\ref{eq:p7} for both hash list and Merkle tree data structures. Moreover, as illustrated in Section~\ref{sec:modelling protocol}, we are also able to prove the security of certificate transparency and transparent decryption with all properties~\ref{eq:p1} to~\ref{eq:p7} declared as axioms. 
All the proofs of the properties in the interface are proved in a single file for hash list in less than a second. For Merkle trees, the proofs are separated in five different files, each taking less than a second to be verified. The proofs of the protocols themselves with the interface are also done in less than two seconds. 

Note that the generalisation of lemmas and axioms can also be of use outside of our methodology. For instance, an earlier version of our work has already been used to prove the Encrypted Client Hello extension of the TLS protocol~\cite{BCW-ccs22}. Specifically, the saturation procedure was entering into a loop when trying to resolve a clause of the form $psk(id,k,c,s) \wedge H \rightarrow \patt(id)$ where $psk(id,k,c,s)$ intuitively represented some shared key $k$ with identity $id$ between a server $s$ and a client $c$. To prevent the loop, they added a lemma indicating that when $psk(id,k,c,s)$ holds then the identity was already known to the attacker or else was honestly generated:

\begin{lstlisting}
lemma id,k,c,s:bitstring;
psk(id,k,c,s) ==> attacker(id) ||
      id = honest_id(s,c,k) [induction].
\end{lstlisting}

The application of the inductive lemma on $psk(id,k,c,s) \wedge H \rightarrow \patt(id)$ would yield two clauses, one with $\blocking{\patt(id)}$ in the hypothesis and one where $id$ is instantiated by $honest_{id}(s,c,k)$. The former would be removed by our new tautology rule, whereas the instantiated second clause would avoid the loop during the saturation.

%% file: app-proverif.tex

\section{Detailed syntax \& semantics}

We reuse most of the notation given in~\cite{BCC-sn22-long}. We simplify however the definitions as we do not consider here equivalence properties, phases or tables.
Recall that $\Cuser$ is the set of Horn clauses given in the input file describing the semantics of the predicates in $\Fp$ and that for all $p \in \Fp$, the semantics of $p$, i.e. $\sem{p}$, is defined by the derivability within $\Cuser$. Additionally, we consider a set of user-given abstract predicates (i.e. predicates declared with the option ``blocking'') $\Fap$ such that for all $p\in \Fap$, the set of $\sem{p}$ is assumed to exist. As mentioned in the body of the paper, we will consider the set $\Fbp$ of $\blocking{p}$ with $p \in \Fp$. Note that given $p \in \Fap$, we define $\blocking{p} = p$ (the predicate is already blocking).


\subsection{Semantics}

\newcommand{\teval}{\Downarrow}

Let $\V$ be an infinite set of variables and $\N$ be an infinite set of names.
The semantics of a destructor function symbol $g$ is given by a sequence of rewrite rules $\defrw(g) = [g(U_{i,1},\ldots,U_{i,n})\rightarrow U_i]_{i=1}^k$ where $U_{i,j}$ and $U_i$ can be either terms or the special constant $\fail$. The result of the evaluation of a term $g(M_1, \ldots,M_n)$ is typically the result of the first rewrite rule in $\defrw(g)$ that matches the arguments $M_1,\ldots, M_n$. Formally, the evaluation of an expression \(D\) to \(U\), denoted \(D \teval U\), is given by:
    \begin{itemize}
      \item if \(D \in \V \cup \N \cup \{\fail\}\) then \(D \teval D\);
      \item if \(D = f(D_1, \ldots, D_n)\) and \(D_1 \teval V_1, \ldots, D_n \teval V_n\):
      \begin{itemize}
        \item if \(f \in \Fc\) and \(\fail \in \{V_1, \ldots, V_n\}\) then \(D \teval \fail\);
        \item otherwise if \(f \in \Fc\) then \(D \teval f(V_1, \ldots, V_n)\);
        \item otherwise \(f \in \Fd\) and we let \(D' = g(V_1, \ldots, V_n)\). 
        Then either \(D'\) cannot be rewritten by any rule of \(\defrw(g)\) and \(D \teval \fail\), or otherwise \(D \teval V\) where \(D' \to V\) by the first rule of \(\defrw(g)\) applicable to \(D'\).
      \end{itemize}
    \end{itemize}

The semantics of processes is given in Fig.~\ref{fig:semantics}.

\begin{figure*}[!ht]
\begin{center}
\[
\begin{array}{@{}lr}
	\E,\Pro \cup \multiset{0},\Att \lstep{}  \E, \Pro,\Att &(\RNil)\\[1mm]
	\E,\Pro \cup \multiset{P \mid Q},\Att \lstep{} \E,\Pro \cup \multiset{P,Q},\Att &(\RPar)\\[1mm]
	 \E, \Pro \cup \multiset{!P},\Att \lstep{} \E, \Pro \cup \multiset{P,!P},\Att &(\RRepl)\\[1mm]
	 \E, \Pro \cup \multiset{\new a;P},\Att \lstep{}  \E \cup \{a'\}, \Pro \cup \multiset{P\{^{a'}/_a\}},\Att\hfill\text{if $a' \not\in \E$}&(\RRestr)\\[1mm]
  \E,\Pro \cup \multiset{\out(N,M);P, \inn(N,x);Q},\Att \lstep{\fmsg{N}{M}} \E, \Pro \cup \multiset{P, Q\{^M/_x\}},\Att&(\RIO)\\[1mm]
	 \E, \Pro,  \Att \lstep{\fmsg{N}{M}}  \E, \Pro,\Att \hfill\text{if $N,M \in \Att$}&(\RMsg)\\[1mm]
	 \E,\Pro \cup \multiset{\llet x=D \inelse P \eelse Q}, \Att \lstep{}  \E,\Pro\cup\multiset{P\{^M/_x\}}, \Att\quad\hfill\text{if $D \teval M$}&(\RLetin)\\[1mm]
	 \E,\Pro \cup \multiset{\llet x=D \inelse P \eelse Q}, \Att \lstep{}  \E,\Pro\cup\multiset{Q}, \Att\quad\hfill\text{if $D \teval \fail$}&(\RLetelse)\\[1mm]
	 \E,\Pro \cup \multiset{\out(N,M);P},\Att \lstep{\fmsg{N}{M}} \E, \Pro \cup \multiset{P},\Att \cup \{M\}\quad\hfill\text{if $N \in \Att$} &(\ROut)\\[1mm]
	 \E,\Pro \cup \multiset{\inn(N,x);Q},\Att \lstep{\fmsg{N}{M}} \E, \Pro \cup \multiset{Q\{^M/_x\}},\Att \quad\hfill\text{if $N,M \in \Att$}&(\RIn)\\[1mm]
	 \E,\Pro,\Att \lstep{} \E, \Pro,\Att \cup \{M\}&(\RApp)\\
	\hfill\text{if $M_1,\ldots, M_n \in \Att$, $f/n \in \Fc \cup \Fd$ and $f(M_1,\ldots,M_n) \teval M$}&\\[1mm]
	 \E,\Pro,\Att \lstep{}  \E \cup \{a'\},\Pro,\Att \cup \{a'\}\quad\hfill\text{if $a' \not\in \E$}&(\RNew)\\[1mm]
	 \E,\Pro \cup \multiset{\event(ev); P},\Att \lstep{\fevent{ev}}  \E,  \Pro \cup \multiset{P},\Att &(\REvent)\\[1mm]
	\E, \Pro \cup \multiset{\llet x_1,\ldots, x_n \suchthat p(M_1,\ldots,M_k) \inelse P \eelse Q} \lstep{} \E,\Pro \cup \{P\sigma\}, \Att &(\RPredin)\\
\hfill\text{if $\dom{\sigma} = \{ x_1,\ldots,x_n\}$ and $p(M_1\sigma,\ldots,M_k\sigma) \in \sem{p}$}&\\[1mm]
	\E, \Pro \cup \multiset{\llet x_1,\ldots, x_n \suchthat p(M_1,\ldots,M_k) \inelse P \eelse Q} \lstep{} \E,\Pro \cup \{Q\}, \Att &(\RPredelse)\\
\hfill\text{if for all $\sigma$ with $\dom{\sigma} = \{ x_1,\ldots,x_n\}$, $p(M_1\sigma,\ldots,M_k\sigma) \not\in \sem{p}$}&\\[1mm]
\end{array}
\]
\end{center}
\caption{Transitions between configurations.}
\label{fig:semantics}
\end{figure*}


\subsection{Satisfaction of facts and queries}

To express secrecy and correspondence properties, ProVerif consider facts as follows:
\[
\begin{array}{rl@{\qquad}l}
F ::= & \fevent{ev(M_1,\ldots,M_n)}&ev\text{ event}\\
      & \fmsg{M}{N}&\\
      & \fatt{M}&\\
      & p(M_1,\ldots,M_n)& p \in \Fp\\
      & M = N&\\
      & M \neq N&\\
\end{array}
\]

The satisfaction relation $\satisfy$ is defined over a trace $T$, a step $\tau$ and a ground fact $F$, denoted $T,\tau \satisfy F$, as follows:
\begin{itemize}
\item $T,\tau \satisfy \fatt{M}$ iff $T[\tau] \lstep{}^* \E, \Pro, \Att$  by only the rules \RApp and \RNew such that $M \in \Att$.
\item $T, \tau \satisfy \fmsg{M}{N}$ iff $T[\tau-1] \lstep{\fmsg{M}{N}} T[\tau]$.
\item $T, \tau \satisfy \fevent{ev}$ iff $T[\tau-1] \lstep{\fevent{ev}} T[\tau]$.
\item $T, \tau \satisfy p(M_1,\ldots,M_n)$ iff $p(M_1,\ldots,M_n) \in \sem{p}$ .with $p \in \Fp \cup \Fap$.
\item $T, \tau \satisfy M = N$ iff $M = N$.
\item $T, \tau \satisfy M \neq N$ iff $M \neq N$.
\end{itemize}
Note that only facts $\fmsg{M}{N}, \fevent{ev}$ and $\fatt{M}$ actually depend on the step $\tau$ and the trace $T$.  Thus, we also write $\satisfy p(M_1,\ldots,M_n)$, $\satisfy M = N$ and $\satisfy M \neq N$. Finally, when the exact value of the step is unnecessary, we can write $T \satisfy F$ instead of $\exists \tau. T,\tau \satisfy F$.
The satisfaction of first order logic formula over facts is naturally defined based on $\satisfy$. 

As mentioned in the body of the paper, by seen a correspondence query $F_1 \wedge \dots \wedge F_n \Rightarrow \psi$ as a first order logic formula $\Psi = \forall \tilde{x}. (F_1 \wedge \dots \wedge F_n \Rightarrow \exists \tilde{y}. \psi)$ with $\tilde{x} = \vars{F_1,\ldots,F_n}$ and $\tilde{y} = \vars{\psi} \setminus \tilde{x}$, its semantics is defined as: for all traces $T$ of $P$, $T \satisfy \Psi$.


\subsection{From operational to instrumented semantics}

Freshness of names from rules \RNew and \RRestr have always been problematic to handle when it comes to automatic verification. In~\cite{BCC-sn22-long}, they deal with the problem by \emph{instrumenting} their syntax and semantics. Typically, replication processes $!P$ will be associated with a session variable $i$, i.e. $!^i P$, and processes $\new k; P$ will be replaced by $P\{ ^{k[\tilde{x}]}/k\}$ where $k[\tilde{x}]$ represents a pattern and $\tilde{x}$ are the input and session variables in the scope of $\new k$. In other words, the name becomes a function that is applied on its environnement at the time its creation. Such instrumentation is shown to be correct \cite[Lemma 3]{BCC-sn22-long} where the satisfaction relation on instrumented facts is denoted $\satisfyI$.
We naturally extend the satisfaction relation to user-defined predicates as follows: For all $p \in \Fp \cup \Fap$, for all $T, \tau$, $T,\tau \satisfyI p(M_1,\ldots,M_n)$ iff $T,\tau \satisfyI \blocking{p}(M_1,\ldots,M_n)$ iff $p(M_1,\ldots, M_n) \in \sem{p}$.

It is in fact shown that we can also restrict the space of traces when verifying a correspondence query, that is by only considering IO compliant traces~\cite[Definition 13]{BCC-sn22-long}. The exact details of the definition is not relevant to this paper but it intuitively restricts the satisfaction of attacker facts to what is exactly known at a current step (and not what the attacker can learn from its knowledge). Similarly, it forces the satisfaction of message fact to go through the attacker and not through an internal communication when the query does not request to prove message facts.
Note that the IO compliance of a trace and the IO compliance of a query is normally defined w.r.t. a phase but as previously mentioned, we simplified the framework by removing tables and phases. 

We consider the same satisfaction relation $\satisfyIO$ that we extend to user-defined predicates as follows:
\begin{itemize}
	\item $T,\tau \satisfyIO p(M_1,\ldots,M_n)$ iff $p(M_1,\ldots,M_n) \in \sem{p}$,  $\tau = \minsize{\Cuser \cup \CblockPred}{p(M_1,\ldots,M_n)}$, and $p \in \Fp$
	\item $T,\tau \satisfyIO \blocking{p}(M_1,\ldots,M_n)$ with $p \in \Fp \cup \Fap$ iff $\tau = 0$ and $p(M_1,\ldots,M_n) \in \sem{p}$.
\end{itemize}
where $\CblockPred = \{ \rightarrow F \mid \pred(F) = p \in \Fbp \cup \Fap \wedge F \in \sem{p}\}$ and $\minsize{\Cuser \cup \CblockPred}{p(M_1,\allowbreak\ldots,M_n)}$ is the size (in term of number of nodes) of the smallest derivation of $p(M_1,\ldots,M_n)$ in $\Cuser \cup \CblockPred$.

Note that contrary to the satisfaction relation $\satisfy$, in $T,\tau \satisfyIO p(M_1,\ldots,M_n)$, the integer $\tau$ provides some information on the deduced predicate, namely that it corresponds to the size of the smallest derivation capable of deriving the fact. Thus, $\tau$ does not have the \emph{step} signification as it is for attacker, message and event facts. Note that the satisfaction is still independent from the trace $T$.

 
The satisfaction of a correspondence query is the same as in~\cite[Definition 10]{BCC-sn22-long}. We recall it in the simplified version where we do not consider injective events nor nested queries, as they are not the main goal of this paper.



\begin{definition}
\label{def:satisfaction query}
Given two satisfaction relations $\vdash_{p}$ and $\vdash_{c}$, given a set of traces $\mathcal{T}$ of some initial instrumented configuration $\C_I$, we say that $\mathcal{T}$ satisfies $\query$ w.r.t. to $\vdash_{p}$ and $\vdash_{c}$, denoted $(\vdash_p,\vdash_c,\mathcal{T})$, when:

For all traces $T \in \mathcal{T}$, for all tuples of integers $\tilde{\tau} = (\tau_1,\ldots,\tau_n)$, for all substitutions $\sigma$, if $T, \tau_i \vdash_{p} F_i\sigma$ for $i = 1,\ldots, n$ then there exists a substitution $\sigma'$ such that $F_i\sigma = F_i\sigma'$ for $i = 1,\ldots, n$  and $T \vdash_c \qconcl\sigma'$.
\end{definition}

We distinguish two satisfaction relations, one for the premises of the query, i.e. $\vdash_p$, and the other for the conclusion of the query, i.e. $\vdash_c$. This separation originally occurred in~\cite{BCC-sn22-long} as a proof technique. But in this case, we will in fact use to show the different way we need to prove lemmas for them to be used correctly.



We obtain the same result linking the restrictions to the main satisfiability relation:


\begin{lemma}[\hspace{1sp}{\cite[Lemma 8]{BCC-sn22-long}}]
\label{lem:instrumented_to_IO_compliant}
Let $\C_I = \rho, P,\Att$ be an initial instrumented configuration. Let $\query$ be an IO-compliant correspondence query such that $\names(\query) \subseteq \dom{\rho}$. We have:
\[
\begin{array}{c}
(\satisfyI,\satisfyI,\tracestep{\C_I}) \models \query \\
\text{if and only if}\\
(\satisfyIO,\satisfyI,\tracestepIO{\C_I}) \models \query
\end{array}
\]
\end{lemma}


\subsection{Proof by induction}

In~\cite{BCC-sn22-long}, a query is proved by induction typically on the size of the trace and the multiset of steps $(\tau_1,\ldots,\tau_n)$ from \Cref{def:satisfaction query}. Moreover, as the proof by induction typically consists on generating a special \emph{inductive lemma}, representing the inductive hypothesis, which would be applied during the saturation procedure, ProVerif currently simplifies the conclusion of the query to remove non-blocking events when generating such inductive lemma. In our case, thanks to our generalization of blocking predicate, we do not need to simplify the conclusion query but we will replace all predicates in the conclusion of the query with their blocking counterpart. However, our inductive hypothesis had to be extended as stated below.


\begin{definition}
\label{def:inductive_hypothesis}
Let $\query$ be a IO-compliant query of the form $\bigwedge_{i=1}^n F_i \Rightarrow \qconcl$. For all traces $T$, for all tuples of steps $\tuplesteps = (\tau_1,\ldots,\tau_k)$ and $\tuplesteps_p = (\tau_{k+1},\ldots,\tau_m)$, $\IndHyp{\query}{T}{\tuplesteps}{\tuplesteps_p}$ holds, if and only if, if $n = m$ and for all $j \leq n$, $\pred(F_j) \in \Fp \Leftrightarrow j\geq k+1$ then 
\begin{quote}
for all substitutions $\sigma$, if $T,\tau_i \satisfyIO F_i\sigma$ for $i = 1, \ldots, n$ then there exists $\sigma'$ such that $F_i\sigma = F_i\sigma'$ for $i = 1, \ldots, n$ and $T \satisfyIO \ToB{\qconcl\sigma'}$. 
\end{quote}
\end{definition}

There are two key differences with \cite[Definition 15]{BCC-sn22-long}:
\begin{inparaenum}[(i)]
\item we separated the steps from standard facts, $\patt$, $\pmsg$ and $\pevent$, with the ones corresponding to predicate from $\Fp$.
\item the satisfaction of the conclusion of the query relies on $\satisfyIO$ instead of $\satisfyI$, the main difference being related to the satisfaction of $\fatt{M}$ facts. Intuitively, $T,\tau \satisfyI \fatt{M}$ holds if the attacker \emph{can} create $M$ from its knowledge at $\tau$, whereas $T,\tau \satisfyIO \fatt{M}$ ensures that the attacker already has $M$ in its knowledge. This distinction matter for the proof of induction has we need to order the facts, so that we can apply our inductive hypothesis on them.
\end{inparaenum}

To state the properties for lemmas and inductive lemmas needed for expressing the soundness of the saturation, we naturally extended the strict order relation $\orderind$ to the lexicographic ordering of triplets $(T,\tuplesteps,\tuplesteps_p)$ where comparison of the size of $T$ is used for the first element, and multiset ordering on natural number is used for the second and third element, i.e. $\tuplesteps$ and $\tuplesteps_p$ are viewed as a multiset of natural number instead of tuples in the order.

The properties on lemmas and inductive lemmas needed for the soundness can thus be summarized as follow.


\begin{definition}[\hspace{1sp}{\cite[Definition 22]{BCC-sn22-long}}]
\label{def:inductive_hypothesis_main}
Let $\C_I$ be an initial instrumented configuration. Let $\Lemma$, $\Lemma_i$ be two sets of lemmas. Let $\tuplesteps$ and $\tuplesteps_p$ be two tuples of steps. Let $T \in \tracestepIO{\C_I}$. We define the predicate $\HypLemmas{T}{\tuplesteps}{\tuplesteps_p}{\Lemma}{\Lemma_i}$ to hold if and only if 
\begin{itemize}
	\item for all $\query \in \Lemma$, $(\satisfyIO,\satisfyIO,\tracestepIO{\C_I}) \models \query$;
	\item for all $\query \in \Lemma_i$, for all $T' \in \tracestepIO{\C_I}$, for all tuples of steps $\tuplesteps'$, $\tuplesteps'_p$, if $(T',\tuplesteps',\tuplesteps'_p) \orderind (T,\tuplesteps,\tuplesteps_p)$ then $\IndHyp{\query}{T'}{\tuplesteps'}{\tuplesteps'_p}$.
\end{itemize}
\end{definition}

As in~\cite{BCC-sn22-long}, the set $\Lemma$ in the above definition represents the lemmas already proved by ProVerif (or the axioms) whereas $\Lemma_i$ are the query that we are currently trying to prove by induction.


\section{The procedure}


\subsection{Generation of Horn clauses}

The generation of Horn clauses is left unchanged from~\cite{BCC-sn22-long} and are defined by the translation $\GenClause{P,\I}{\Hfacts}{\rho}$ where $P$ is a process, $\I$ is a sequence of session variables and terms, $\Hfacts$ is a conjunction of facts and disequalities for the form $\forall \tilde{x}. M \neq N$, and $\rho$ is a mapping from names to instrumented names. We present in Fig.~\ref{fig:clause generation} the translation fitting our simplified framework. 

\begin{figure}[ht]
\[
\begin{array}{l}
\GenClause{0,\I}{\Hfacts}{\rho} = \emptyset\\[1mm]
\GenClause{P \mid Q,\I}{\Hfacts}{\rho} = \GenClause{P,\I}{\Hfacts}{\rho} \cup \GenClause{Q,\I}{\Hfacts}{\rho}\\[1mm]
\GenClause{! P,\I}{\Hfacts}{\rho} = \GenClause{P,(\I,i)}{\Hfacts}{\rho} \\
\quad\text{where $i$ is a fresh variable session identifier}\\[1mm]
\GenClause{\new a; P,\I}{\Hfacts}{\rho} = \GenClause{P,\I}{\Hfacts}{(\rho[a \mapsto a[\I]])}\\[1mm]
\GenClause{\inn(M,x);P,\I}{\Hfacts}{\rho} = \\
\quad\GenClause{P,(\I,x')}{(\Hfacts \wedge \pmsg{}(M\rho,x'))}{(\rho[x\mapsto x'])}\\[1mm]
\GenClause{\out(M,N);P,\I}{\Hfacts}{\rho} = \\
\quad\GenClause{P,\I}{\Hfacts}{\rho} \cup \{ \Hfacts \rightarrow \pmsg{}(M\rho,N\rho)\}\\[1mm]
\GenClause{\llet x=D \inelse P \eelse Q,\I}{\Hfacts}{\rho} =\\
\quad \bigcup \{ \GenClause{P,\I\sigma}{(\Hfacts\sigma \wedge \phi)}{(\rho\sigma[x\mapsto M])} \mid D\rho \tevalp (M,\sigma,\phi)\}\\
\quad \cup \bigcup \{ \GenClause{Q,\I\sigma}{(\Hfacts\sigma \wedge \phi)}{(\rho\sigma)} \mid D\rho \tevalp (\fail,\sigma,\phi)\}\\[1mm]
\GenClause{\event(ev);P,\I}{\Hfacts}{\rho} = \\
\quad\GenClause{P,\I}{\left(\Hfacts\wedge \psevent(ev\rho)\right)}{\rho} \cup \{ \Hfacts \rightarrow \pmevent(ev\rho)\}\\[1mm]
\GenClause{\llet x_1,\ldots, x_n \suchthat pred \inelse P \eelse Q,\I}{\Hfacts}{\rho} = \\
\quad \GenClause{P,\I}{(\Hfacts \wedge pred\rho')}{(\rho\rho')} \cup \GenClause{Q,\I}{\Hfacts}{\rho}\\
\quad\text{where $\rho' = [x_i \mapsto x'_i]_{i=1}^n$ and $x'_i$ fresh variables.}\\ 
\end{array}
\]
\caption{Generation of clauses}
\label{fig:clause generation}
\end{figure}

Given an initial instrumented configuration $\C_I$, the set of clauses generated from it is denoted $\Cprotocol(\C_I)$. 
We also give back the set of clauses generated for the attacker, denoted $\Cattacker(\C_I)$, as follows:
\begin{align}
&\quad\text{For all $a \in \Att(\C_I)$}, \rightarrow \patt(a[])\label{RInit}\tag{RInit}\\
&\quad\text{$\patt(b_0[i])$}\hfill \text{\quad with $b_0$ not occuring in $\C_I$}\label{RGen}\tag{RGen}\\
&\quad\text{$\patt(\fail)$}\label{RFail}\tag{RFail}\\
\begin{split}
&\quad\text{For all functions $h$,}\\*
&\quad\quad\text{for all $h(U_1,\ldots, U_n)\rightarrow U \with \phi$ in $\defrw(h)$}\\*
&\quad\quad \patt(U_1) \wedge \dots \wedge \patt(U_n) \wedge \phi \rightarrow \patt(U)\\*
\end{split}\tag{Rf}\label{ruleRf}\\
&\quad\pmsg(x,y) \wedge \patt(x) \rightarrow \patt(y)\label{Rl}\tag{Rl}\\
&\quad\patt(x) \wedge \patt(y) \rightarrow \pmsg(x,y)\label{Rs}\tag{Rs}
\end{align}



Given a trace $T$ of an initial instrumented configuration $\C_I$, we denote by $\CblockT{T}$ the set of Horn clauses of the form $\rightarrow \blocking{p}(M_1,\ldots,M_n)$ such that $T,\tau \satisfyIO \blocking{p}(M_1,\ldots,M_n)$ for some step $\tau$. Note that $\CblockT{T}$ contains also $\CblockPred$.


\subsection{Derivation and initial soundness}

We preserve exactly the notion of derivation of Horn clauses used in the saturation but the satisfaction relation of a derivation by a trace had to be extended.


\begin{definition}[\hspace{1sp}{\cite[Definition 17]{BCC-sn22-long}}]
\label{def:derivation}
Let $\Cset$ be a set of clauses. Let $F$ be a closed fact and a step $\tau$. A derivation $\D$\index{derivation} of $F$ at step $\tau$ from $\Cset$ is a finite tree defined as follows:
\begin{itemize}
	\item its nodes (except the root) are labelled by clauses from $\Cset$.
	\item its edges are labelled by ground facts and by a step.
	\item if the tree contains a node labelled by $R$ with one incoming edge labelled by $F_0$ and n outgoing edges labelled by $F_1, \ldots, F_n$ then $R \subsume F_1 \wedge \dots \wedge F_n \rightarrow F_0$.
	\item the root has one outgoing edge, labelled by $F$ and $\tau$.
\end{itemize}
\end{definition}

In the above definition, $R \subsume R'$ denotes the multiset subsumption of $R'$ by $R$.

We define the \emph{blocking-free size of a derivation $\D$} as the number of nodes in $\D$ not labelled by a clause of the form $\rightarrow C$ with $\pred(C) = \blocking{p}$ for some $p$. The notion of satisfaction of a derivation by a trace is also parametrized by a set of predicate $\allowedpreds$ on which we need to have guarantees on their satisfiability in the trace w.r.t. $\satisfyIO$. Typically, $\allowedpreds$ will contain the predicate $\patt$, $\pmsg$ and $p \in \Fp$ if they occur in the conclusion of the query or in the premise of a lemma. It will also contain all the blocking predicates and the user-defined predicate $\Fap$.


\begin{definition}
\label{def:satisfy_derivation}
Let $\D$ be a derivation of $F$ at step $\tau$. Let $\allowedpreds$ be a set of predicates. We say that 
	\emph{a trace $T$ satisfies a derivation $\D$ w.r.t. $\allowedpreds$}, denoted $T, \allowedpreds \satisfyD \D$, when for all nodes $\eta$ of $\D$, if $F_0,\tau_0$ is the label of the incoming edge of $\eta$ then 
\begin{enumerate}
	\item if $\eta$ is labelled with \eqref{Rl} then the outgoing edges of $\eta$ are labelled $\pmsg(N,M), \tau'$ and $\patt(N), \tau''$ for some $N,M,\tau',\tau''$ such that $\tau' \leq \tau_0$ and if $\patt \in \allowedpreds$ then $\tau'' < \tau_0$ else $\tau'' \leq \tau_0$.\label{enum-satisderiv-Rl}
	\item if $\eta$ is not labelled with \eqref{Rl}, then for all outgoing edges of $\eta$ labelled $F',\tau'$,  if $\pred(F_0) \in \Fp$ or $\pred(F') \not\in \Fp$ then
		\begin{itemize}
		\item $\pred(F') \not\in \allowedpreds$ implies $\tau' \leq \tau_0$
		\item $\pred(F') \in \allowedpreds$ implies $\tau' < \tau_0$
		\end{itemize}
		\label{enum-satisderiv-rules}
	\item if $F_0 = \patt(f(M_1,\ldots,M_m))$ with $f \in \Fdata$ then
		\begin{itemize}
		\item either $R$ is the clause $\patt(x_1) \wedge \dots \wedge \patt(x_m) \rightarrow \patt(f(x_1,\allowbreak\ldots,x_m))$ and if $\patt \in \allowedpreds$ then $\tau_0 = \min\{ \tau \mid T,\tau \satisfyIO F_0\}$;
		\item or $\eta$ is not the root and the node $\eta'$ connected to the incoming edge of $\eta$ is labelled with the clause $\patt(f(x_1,\ldots,x_m)) \rightarrow \patt(x_j)$ for some $j$ and if $\patt \in \allowedpreds$ then $f(M_1,\ldots,M_m) \in \Att(T[\tau_0])$;
		\end{itemize}
		else $\pred(F_0) \in \allowedpreds$ implies $\tau_0 = \min\{ \tau \mid T,\tau \satisfyIO F_0\}$\label{enum-satisderiv-satisfy}
	\end{enumerate}
\end{definition}

The main different before our definition and~\cite[Definition 18]{BCC-sn22-long} lies in \Cref{enum-satisderiv-rules} where we relaxed the ordering conditions between incoming and outgoing edges, but only in the particular case where the incoming edge is labelled by standard predicates, $\patt$ or $\pmsg$, and the outgoing edges is labelled by a predicate fact $F$ with $\pred(F)$ being a user-defined predicate. The clause labeling such a node typically only comes from the translation of the process into Horn clauses (recall that $\Cuser$ cannot contain $\patt$ nor $\pmsg$). However, in the satisfaction relation $T,\tau \satisfyIO F$, the step $\tau$ represents an actual step of the trace $T$ when $\pred(F)$ are standard predicates whereas $\tau$ represents the size of the smallest derivation in $\Cuser$ when $\pred(F) \in \Fp$. These two notions being unrelated, we cannot hope to order them. When $\pred(F) \in \Fap \cup \Fbp$, the satisfaction $T,\tau \satisfyIO F$ implies $\tau = 0$ and so we are naturally ordered.

Note, on the order hand, when both incoming and outgoing edge are labelled by user-defined predicates (i.e. the Horn clause labeling the node originated from $\Cuser$), we do require that the step $\tau_0$ of the incoming edge is strictly bigger than the one of the outgoing edge (when required by $\allowedpreds$). This will typically allow us to apply inductive hypotheses from \Cref{def:inductive_hypothesis_main}.

Another difference comes from \Cref{enum-satisderiv-satisfy} where we require $\tau_0$ to be the minimum step that satisfies $F_0$. We can basically always take such minimum by replacing the corresponding sub derivation with one proving the same fact with a smaller step. It is however useful when relating standard predicates and their blocking counterpart.

With this, we can state the first soundness result linking the satisfaction of facts in a trace and the satisfaction of a derivation. The following theorem is an adaption of \cite[Theorem 1]{BCC-sn22-long}.


\begin{theorem}
\label{lem:soundness_correspondence_query}
Let $\C_I$ be an initial instrumented configuration. Let $\allowedpreds$ be a set of predicates.

For all $T \in \tracestepIO{\C_I}$, for all ground facts $F$, for all steps $\tau$, if $T, \tau \satisfyIO F$ then there exists a derivation $\D$ of $F$ at step $\tau$ from $\Cattacker(\C_I) \cup \Cprotocol(\C_I) \cup \CblockT{T}$ such that $T, \allowedpreds \satisfyD \D$.
\end{theorem}

\begin{proof}[Proof sketch]
The proof is typically the same as in~\cite[Theorem 1]{BCC-sn22-long} as we did not modify the core generation of Horn clauses and satisfiability relations but only extended them user-defined predicates. Hence we only need to discuss the new cases arose by them.

Notice that if $\pred(F)$ is a blocking predicate $\blocking{p}$, $T, \tau \satisfyIO F$ directly gives us that $\rightarrow F$ is in $\CblockT{T}$ which allows us to conclude.

In the case where the protocol executes the construct $\llet x_1,\ldots,\allowbreak x_n \suchthat \allowbreak p(M_1,\ldots,M_n) \inelse P \eelse Q$, the semantics rule (\textsc{Pred1}) ensures that $p(M_1\sigma,\ldots,\allowbreak M_n\sigma)$ is derivable from $\Cuser \cup \CblockPred$. 
Thus to satisfy \Cref{enum-satisderiv-rules} of \Cref{def:satisfy_derivation}, we only need to take the smallest derivation $\D$ of $p(M_1,\ldots,M_n)\sigma$ from $\Cuser \cup \CblockPred$.

Notice that for the semantics rule (\textsc{LetP2}) requires for the predicate to not hold; however in the translation of the construct to Horn clauses, this condition is "lost". It is one of the strong abstraction of ProVerif: it behaves as if the else branch of a $\llet \ldots \suchthat$ construct can always be executed. This is due to the fact that ProVerif does not reason with the negation of predicates.

When $\pred(F) \in \Fp$, we once again know that $F \in \sem{p}$ which by definition implies that $F$ is derivable from $\Cuser \cup \CblockPred$. By taking the smallest derivation, we can conclude.
\end{proof}


\subsection{Saturation procedure and its soundness}

As previously mentioned in the body of this paper, the selection function is extended to prevent the selection of any blocking predicate facts (including the ones from $\Fap$) and attacker facts $\fatt{x}$ with $x$ a variable.

We already gave some details on how to apply lemma in our new settings but we provide the complete definition here, as well as the application of inductive lemmas and some updated classic transformation rules.


\paragraph{Lemmas} 
Let us denote $\Cstd$ the set of clauses containing the clauses \eqref{Rl}, $(\text{Rf}_{g})$ and $(\text{Rf}_{\pi^g_i})$ for all $g \in \Fdata$ and $i$. The following transformation rule (\Rlem{\Lemma,\allowedpreds}) describes the application of (non-inductive) lemmas in $\Lemma$ with the set of predicate $\allowedpreds$.

\begin{prooftree}
	\AxiomC{
	\begin{tabular}{@{}c@{}}
	$\Cset \cup \{ R=(H \rightarrow C)\}$\qquad
	$R \not\in \Cstd$\\
	$(\bigwedge_{i=1}^n F_i \Rightarrow \bigvee_{j=1}^m \qconcl_j) \in \Lemma$\\
	$\forall i,\pred(F_i) \in \allowedpreds\text{ and either }\ToB{F_i\sigma} \in \ToB{H}$\\
	or ($F_i\sigma = C$ and $\forall j,\forall F \in \qconcl_j$, $\mgu{F\sigma,C} = \bot$)\\
	or ($F_i\sigma = C$ and $\pred(C) \in \Fp$)
	\end{tabular}
	}
	\UnaryInfC{$\Cset \cup \{ H \wedge \ToB{\qconcl_j\sigma} \rightarrow C \}_{j=1}^m$}
\end{prooftree}

Proved lemmas can be applied on both hypotheses and the conclusion. However, when applied thanks to the conclusion, we want to preserve the invariant that the facts in the hypotheses of a Horn clause occurs strictly before the conclusion of the clause. However, a lemma only guarantees us that the facts in $\bigvee_{j=1}^m \qconcl_j$ occur before or at the same time as $F_i$. Hence the condition on unification of facts from the conclusion of the query with the conclusion of the clause. Note that this does not apply when the conclusion of the clause is rooted by a user-defined predicate as a lemma only adds blocking predicates from $\Fbp \cup \Fap$ which by definition of the satisfaction relation $\satisfyIO$ are always satisfied at step $0$. 


\paragraph{Inductive lemmas}
The application inductive lemmas differ from normal lemmas by the fact that we do not try to match the conclusion of the Horn clause with the premise of the inductive lemma. Recall that an inductive lemma intuitively represents the inductive hypothesis. Hence we can only apply it on facts occurred strictly before the conclusion. For the same reason, we also do not match the premise of the inductive lemma with a fact rooted by a blocking predicate $\blocking{p}$ with $p \in \Fp$. Indeed, when such predicate occurs in the hypothesis of a clause, it means that it was added by a previous application of a lemma (inductive or not). However, our induction for user-defined predicate is done typically on the size of the minimal derivation in $\Cuser \cup \CblockPred$. Queries and lemmas do not guarantee any relation on the size of derivations between facts of their premise and their conclusion. Hence, we must not apply our inductive lemma on these facts. The following transformation rule (\Rind{\Lemma_i,\allowedpreds}) describes the application of inductive lemmas from $\Lemma_i$ with the set of predicates $\allowedpreds$.

\begin{prooftree}
	\AxiomC{
	\begin{tabular}{c}
	$\Cset \cup \{ R=(H \rightarrow C)\}$\qquad
	$R \not\in \Cstd$\\
	$(\bigwedge_{i=1}^n F_i \Rightarrow \bigvee_{j=1}^m \qconcl_j) \in \Lemma_i$\\
	$\forall i,\pred(F_i) \in \allowedpreds$ and either $F_i\sigma \in H$\\
	or ($\pred(F_i) \not\in \Fp$ and $\ToB{F_i\sigma} \in \ToB{H})$
	\end{tabular}
	}
	\RightLabel{}	\UnaryInfC{$\Cset \cup \{ H \wedge \ToB{\qconcl_j\sigma} \rightarrow C \}_{j=1}^m$}
\end{prooftree}


\paragraph{Classic transformation rule}
As previously mentioned in the body of this paper, we can amend some of the classic simplification rules to benefit from the blocking facts introduced by lemmas. Note that all these extended simplifications rules are sound but it does not mean that they should always be applied. Depending on the needs (precision, termination, \ldots), one may prefer to rely on the usual rules.

Consider for the tautology rule:

\begin{prooftree}
	\AxiomC{
	\begin{tabular}{c}
	$\Cset \cup \{ F' \wedge H \rightarrow F\}$\\
	$F = F'$ or ($\pred(F) \not\in \Fp$ and $\ToB{F} = \ToB{F'}$)
	\end{tabular}
	}
	\UnaryInfC{$\Cset$}
\end{prooftree}

This rules does not make the distinction between facts and their blocking counter fact when the predicate rooting $F$ is standard, i.e. not in $\Fp$. Indeed, for these predicates, the application of a lemma preserves the ordering, i.e. we know that $F'$ occurred strictly before $F$ in the sense of trace steps. When $F$ is rooted by $\Fp$ however, its blocking counterpart is not ordered with respect to the conclusion and can be seen as "additional information" that can be propagated later. As such, it does not necessary implies that such clause with $F' = \ToB{F}$ and $\pred(F) \in \Fp$ is useless.  

We update the redundant hypotheses rule as follows.
\begin{prooftree}
	\AxiomC{
	\begin{tabular}{c}
	$\Cset \cup \{ H' \wedge H \wedge \phi \rightarrow C\}$\quad
	$\ToB{H'\sigma} \subseteq \ToB{H}$\\
	$\dom{\sigma} \cap \vars{H,C} = \emptyset$\quad
  $\phi \models \phi\sigma$
	\end{tabular}
	}
	\UnaryInfC{$\Cset \cup \{ H \wedge \phi\sigma \rightarrow C\}$}
\end{prooftree}

We also update the attacker rule:
\begin{prooftree}
	\AxiomC{
	\begin{tabular}{c}
	$\Cset \cup \{ F \wedge H \wedge \phi \rightarrow C\}$\quad
	$\ToB{F} = \blocking{\patt}(x)$\\
	$x$ does not appear in $H$, $C$
	\end{tabular}
	}
	\UnaryInfC{$\Cset \cup \{ H \wedge \phi \rightarrow C\}$}
\end{prooftree}


\paragraph{Subsumption}
We can also adjust the definition of subsumption of clauses to make use of blocking predicates. We expect this modification to have a significant positive impact on the termination of ProVerif but it may also have a negative impact on its precision.


\begin{definition}
\label{def:subsumption}
Let $H_1 \wedge \phi_1 \rightarrow C_1$ and $H_2 \wedge \phi_2
\rightarrow C_2$ be two clauses. We say that $H_1 \wedge \phi_1 \rightarrow C_1$ subsumes $H_2 \wedge \phi_2 \rightarrow C_2$, denoted $(H_1 \wedge \phi_1 \rightarrow C_1) \subsume_{\mathsf{b}} (H_2\wedge \phi_2 \rightarrow C_2)$, when there exists $\sigma$, $n, F_1,\dots, F_n, G_1,\dots, G_n$ such that 
	\begin{itemize}
	\item $C_1\sigma = C_2$
	\item $H_2 = G_1 \wedge \dots \wedge G_n \wedge H_2'$
	\item $H_1 = F_1 \wedge \dots \wedge F_n$
	\item for all $i \in \{1,\ldots,n \}$, either $F_i\sigma = G_i$ or $F_i\sigma = \ToB{G_i}$
	\item $\phi_2 \models \phi_1\sigma$.
	\end{itemize}
\end{definition}

Notice that a fact rooted by $\blocking{p}$ can subsume a fact rooted by $p$. However, the converse is not true, that is a fact rooted by a non-blocking predicate $p$ cannot subsume its blocking counterpart $\blocking{p}$.


\paragraph{Soundness of the saturation}
For the rest of the saturation procedure, we proceed as in the current version of ProVerif. In particular, all other simplification rules remain unchanged (for natural number, data constructor, etc) as well as the general redundancy~\cite[Section 5.3]{BCC-sn22-long}.

We say that a set of clauses $\Cset$ is consistent with $\Fp$ when for all $(H \rightarrow C) \in \Cset$, 
\begin{enumerate}
  \item $\pred(C) \in \Fap \cup \Fbp$ implies $H = \emptyset$;
  \item $\pred(C) \in \Fp$ implies for all $F \in H$, $\pred(F) \in \Fp \cup \Fap \cup \Fbp$.
\end{enumerate}
This property is directly satisfied by $\Cuser$ and the initial clauses generated from the protocol. It is also preserved by the saturation procedure.

\begin{theorem}
\label{th:saturation soundness}
Let $\C_I$ be an initial instrumented configuration. Let $\allowedpreds$ be a set of predicates. Let $\Lemma$, $\Lemma_i$ be two sets of lemmas. Let $T \in \tracestepIO{\C_I}$.

For all set of clauses $\Cset$ consistent with $\Fp$ containing $\Cstd$, for all derivations $\D$ of $F$ at step $\tau$ from $\Cset \cup \CblockT{T}$, if $T, \allowedpreds \satisfyD \D$ and
\begin{itemize}
	\item $\pred(F) \in \Fp$ implies $\HypLemmas{T}{()}{(\tau)}{\Lemma}{\Lemma_i}$
	\item $\pred(F) \not\in \Fp$ implies $\HypLemmas{T}{(\tau)}{()}{\Lemma}{\Lemma_i}$
\end{itemize}
then there exists a derivation $\D'$ of $F$ at step $\tau$ from $\saturate{\Cset}{\allowedpreds}{\Lemma}{\Lemma_i} \cup \CblockT{T}$ such that $T,\allowedpreds \satisfyD \D'$.
\end{theorem}

\begin{proof}[Proof sketch]
The proof follows the proof of Theorem~2 of~\cite{BCC-sn22-long}. The main idea is to first show that if the derivation $\D$ of $F$ uses a Horn clauses that is transformed by one of the transformation rule, then we can build a new derivation of $\D'$ of $F$ that relies on the Horn clauses produced by the transformation rule. 
The main different is of course the addition of the additional conditions related to blocking predicates but the argument necessary for the proof are similar to the current proof. For instance, for the tautology, as previously mentioned, if $\D$ was using a clause $F' \wedge H \rightarrow F''$ with $\pred(F'') \not\in \Fp$ and $\ToB{F''} = \ToB{F'}$, we would deduce from \Cref{def:satisfy_derivation} that there exists $\tau'$ and $\tau''$ such that $T, \tau' \satisfyIO F'$ and $T, \tau'' \satisfyIO F''$ and $\tau' < \tau''$. Moreover, $\tau'$ and $\tau''$ should be the minimum values that satisfy $F'$ and $F''$ respectively (see \Cref{enum-satisderiv-satisfy}). However, as $\ToB{F''} = \ToB{F'}$, we can deduce by definition of the satisfaction relation $\satisfyIO$ that $\tau' = \tau''$ which is a contradiction with $\tau' < \tau''$. Thus, $\D$ cannot use the clause $F' \wedge H \rightarrow F''$ and hence we can discard the clause. 

Adapting the proofs for the other transformation rules follows the same kind of arguments.
\end{proof}


\subsection{Solving procedure.}

The verification of a correspondence query in ProVerif relies on an extended notion of Horn clauses, namely \emph{ordered Horn clauses}, where:
\begin{itemize}
	\item The conclusion of the Horn clause is a special fact representing a conjunction of facts. Typically, $\bigcurlywedge_{i=1}^n F_i$ represents the fact corresponding to the conjunction $\bigwedge_{i=1}^n F_i$. We refer the reader to section 6.1 of ~\cite{BCC-sn22-long} for more details.
	\item The hypothesis of the Horn clauses can be ordered with respect to each fact $F_i$ of the conclusion using a partial ordering function.
\end{itemize}


\begin{definition}[\hspace{1sp}{\cite[Definition 24]{BCC-sn22-long}}]
\label{def:order function and sets}
We call \emph{ordering function} a partial function $\delta : \mathbb{N} \rightarrow \{ \less, \lesseq \}$. We call \emph{ordered facts}, denoted $F^\delta$, a fact $F$ annotated with an ordering function $\delta$. Finally, an \emph{ordered clause}\index{ordered clause} is a clause of the form $H \rightarrow \bigcurlywedge_{i=1}^n F_i$ where each conjunct in $H$ is an ordered fact.
\end{definition}

The notion of derivation of a conjunction fact $\bigcurlywedge_{i=1}^n F_i$ and is satisfiability by a trace and a tuple of steps is extended in \cite[Definition 26]{BCC-sn22-long}. These definitions remain unchanged. In fact the only part that changes in our setting is the resolution rule where we need to take into account the changes in \Cref{enum-satisderiv-rules} of \Cref{def:satisfy_derivation}, as well as the rules for applying lemmas and inductive lemmas.


\begin{definition}
Let $\allowedpreds$ be a set of predicates. We say that an ordered clause $F_1^{\delta_1} \wedge \dots \wedge F_n^{\delta_n} \wedge \phi \rightarrow C$ satisfies $\allowedpreds$ when for all $i \in \{1,\ldots, n\}$, if $\pred(F_i) \not\in \allowedpreds$ then $\less \not\in img(\delta_i)$.
\end{definition}


\paragraph{Resolution rule}
The main difference between the resolution rule during the saturation procedure and the verification procedure in~\cite{BCC-sn22-long} comes from the assignment of ordering function on the hypotheses of the clauses from the saturated set. In \cite{BCC-sn22-long}, they relied on the invariant that all facts rooted by a predicate in $\allowedpreds$ occurred strictly before the conclusion. In our case, this is not always the case (as indicated in \Cref{enum-satisderiv-rules} of \Cref{def:satisfy_derivation}) when the hypothesis is rooted by a user-defined predicate and the conclusion of the Horn clause is rooted by a standard predicate. In that case, the ordering function associated to this fact in the hypothesis should be empty.

Given an ordering function $\delta$, we denote by $\strict{\delta}$ the partial function with the same domain as $\delta$ and such that $\strict{\delta}(i) = \less$ for all $i \in \dom{\delta}$. 
Moreover, given two facts $F,C$ and an ordering function $\delta$, we denote by $\delta_{res}(F,C,\delta)$ the ordering function defined as:
\begin{itemize}
	\item $\delta_{res}(F,C,\delta) = \emptyset$ if ($\pred(F) = \blocking{p}$ and $p \in \Fp\cup \Fap$) or ($\pred(F) \in \Fp$ and $\pred(C) \not\in \Fp$)
	\item otherwise $\delta_{res}(F,C,\delta) = \strict{\delta}$ if $\pred(F) \in \allowedpreds$
	\item otherwise $\delta_{res}(F,C,\delta) = \delta$
\end{itemize}
Typically, $C$ represents the conclusion of the saturated Horn clause, $F$ is a fact from its hypotheses and $\delta$ is the ordering function of ordered fact on which we will apply the resolution. The first bullet point covers the case where $F$ is a user-defined block predicate as we always assign an empty ordering function on them as well as the case where $F$ is a user-defined predicate derived from the protocol (correspond the negation of the condition in \Cref{enum-satisderiv-rules} of \Cref{def:satisfy_derivation}). The second and third bullet point are respectively the consequence of the second and first bullet point in \Cref{enum-satisderiv-rules} of \Cref{def:satisfy_derivation}.

The resolution rule (\RresOrd{\allowedpreds}) is therefore defined as follows:

 \begin{prooftree}
 	 \AxiomC{
 	 	\begin{tabular}{@{}c@{}}
 	 	$F_1 \wedge \dots \wedge F_n \rightarrow C \in \Csat$\quad
 	 	$\sigma = \mgu{F,C}$\\
 	 	$R = (F^\delta \wedge H_o \rightarrow C') \in \Cset$
 	 	 \quad $F \in \select(R)$\\
 	 	for all $i$, $\delta_i = \delta_{res}(F_i,C,\delta)$
 	 	\end{tabular}
 	 }
     \UnaryInfC{$F_1^{\delta_1}\sigma \wedge \dots \wedge F_n^{\delta_n}\sigma \wedge H_o\sigma \rightarrow C'\sigma$}
 \end{prooftree}


\paragraph{Lemma simplification rule}
As explained in~\cite{BCC-sn22-long}, the application conditions for a lemma during the verification procedure are basically the same as in the saturation procedure. The only difference comes from the assignment of ordering functions to the facts of the conclusion of the query. In~\cite{BCC-sn22-long}, such ordering function was applied to all facts since they were events. In our case however, we need to distinguish between standard predicates and user-defined predicates. As previously mentioned, the blocking counterpart of standard predicates can be ordered whereas the ones of user-defined predicates cannot. We take this into the transformation $\ToB{\qconcl}_\delta$. Formally, $\ToB{\qconcl}_\delta$ is the conjunction built from $\qconcl$ where all facts $F$ is replaced by  $\ToB{F}_\delta$ when $\pred(F) \not\in \Fp$ and by $\ToB{F}_\emptyset$ otherwise (here $\emptyset$ refer to the ordering function with empty domain).

We now define what it means for a lemma to match an ordered Horn clause: Given an ordered clause $R = (H \rightarrow \bigcurlywedge_{i=1}^n F'_i)$ and a lemma $\query = (\bigwedge_{i=1}^m F_i \Rightarrow \bigvee_{j=1}^k \qconcl_j)$, we say that $\query$ matches $R$ with ordering function $\delta_1,\ldots, \delta_m$ and substitution $\sigma$ when for all $i \in \{1,\ldots, m\}$, $\pred(F_i) \in \allowedpreds$ and 
\begin{itemize}
	\item either there exists $F''$ such that $F''^{\delta_i}_i \in H$ and $\ToB{F''_i} = \ToB{F_i\sigma}$
	\item or there exists $j \in \{1,\ldots, n\}$ such that $F_i\sigma = F'_j$ and $\delta_i = \{ j \mapsto \lesseq \}$.
\end{itemize}

The application rule (\RlemOrd{\Lemma,\allowedpreds}) is thus defined as follows:

\begin{prooftree}
	\AxiomC{\begin{tabular}{c}
		$\Cset \cup \{ R = (H \rightarrow \bigcurlywedge_{i=1}^n F'_i)\}$ \\
		$\query = (\bigwedge_{i=1}^m F_i \Rightarrow \bigvee_{j=1}^k \qconcl_j) \in \Lemma$\\
		$\query$ matches $R$ with $\delta_1,\ldots, \delta_m$ and $\sigma$\\
		for all $i$, $\pred(F_i) \not\in \Fp$ implies $\delta \subsume_o \delta_i$
	\end{tabular}
	}
	\UnaryInfC{$\Cset \cup \{ H \wedge \ToB{\qconcl_k\sigma}_\delta \rightarrow \bigcurlywedge_{i=1}^n F'_i\}_{j=1}^k$}
\end{prooftree}


\paragraph{Inductive lemma simplification rule}
Applying inductive lemmas consists basically to satisfy the same application conditions as for proved lemmas but we an additional condition showing that we satisfy the order $\orderind$, i.e. $(T,\tuplesteps',\tuplesteps'_p) \orderind (T,\tuplesteps,\tuplesteps_p)$. This can be achieved by analysing the ordering functions in the hypotheses of the clauses.


\begin{definition}
\label{def:equal_strict_ordering_function}
Let $n_a \leq n_b \in \mathbb{N}$. Let $\delta_1,\ldots,\delta_m$ ordering functions with $\dom{\delta_i} \subseteq \{n_a,\ldots,n_b\}$ for all $i \in \{1,\ldots, m\}$. We say that \emph{$\delta_1,\ldots, \delta_m$ are $(n_a,n_b)$-equal for $\{ j_1,\ldots, j_m\} \subseteq \{ n_a,\ldots,n_b\}$} when: $m \leq n_b - n_a + 1$ and
	\begin{enumerate}
		\item for all $k \in \{1,\ldots,m\}$, $\delta_k(j_k)$ is defined;\label{enum-strict-defined}
		\item for all $k,k' \in \{1,\ldots, m\}$, $k \neq k'$ implies $j_k \neq j_{k'}$\label{enum-strict-distinct}
	\end{enumerate} 
We say that \emph{$\delta_1,\ldots, \delta_m$ are $(n_a,n_b)$-strict for $\{ j_1,\ldots, j_m\} \subseteq \{ n_a,\ldots,n_b\}$} when $n_a \leq n_b$ and:
\begin{itemize}
	\item either for all $k \in \{1,\ldots,m\}$, there exists $i \in \{n_a,\ldots,n_b\}$ such that $\delta_k(i) = \less$
	\item or $\delta_1,\ldots, \delta_m$ are $(n_a,n_b)$-equal for $\{ j_1,\ldots, j_m\}$ and either $m \leq n_b - n_a$ or there exists $k \in \{1,\ldots,m\}$ such that $\delta_{k}(j_{k}) = \less$.
\end{itemize}
\end{definition}


\begin{lemma}
\label{lem:n_strict_ordering_function}
Let $n_a \leq n_b \in \mathbb{N}$. Let $\delta_1,\ldots,\delta_m$ be ordering functions. Let $\tau_{n_a},\ldots,\tau_{n_b}$ and $\tau'_1,\ldots,\tau'_m$ steps. Assume that for all $i \in \{n_a,\ldots,n_b\}$, for all $j \in \{1,\ldots,m\}$, $\delta_j(i)$ defined implies $\tau'_j \mathrel{\delta_j(i)} \tau_i$. 

We have: If there exists $\{ j_1,\ldots, j_m\} \subseteq \{ n_a,\ldots,n_b\}$ such that $\delta_1,\ldots,\delta_m$ are $(n_a,n_b)$-equal (resp $(n_a,n_b)$-strict) for $\{ j_1,\ldots, j_m\}$ then $\multiset{\tau'_1,\ldots,\tau'_m} \leq_m \multiset{\tau_1,\ldots,\tau_n}$ (resp. $<_m$).
\end{lemma}

\begin{proof}
The proof is a reorganisation of the proof of~\cite[Lemma 14]{BCC-sn22-long}. However, since we extended the notion and is one of the fundamental result to apply inductive lemmas, we show the proof again here.

\medskip

Consider first the case where $\delta_1,\ldots,\delta_m$ are $(n_a,n_b)$-equal for $\{ j_1,\ldots, j_m\}$. Thanks to \Cref{enum-strict-distinct}, we know that the indices $j_1,\ldots, j_m$ are distinct and by \Cref{enum-strict-defined} that $\delta_k(j_k)$ is defined for all $k \in \{1,\ldots,m\}$. By being defined, we deduce that for all $k \in \{1,\ldots, m\}$, $\tau'_k \mathrel{\delta_k(j_k)} \tau_{j_k}$ meaning that $\tau'_k \leq \tau_{j_k}$. Hence, we have shown that the steps $\tau'_1,\ldots,\tau'_m$ are smaller than $m$ distinct steps of the multiset $\multiset{\tau_{n_a},\ldots,\allowbreak\tau_{n_b}}$, which entails $\multiset{\tau'_1,\ldots,\tau'_m} \leq_m \multiset{\tau_{n_a},\ldots,\tau_{n_b}}$.

When $\delta_1,\ldots,\delta_m$ are $(n_a,n_b)$-strict for $\{ j_1,\ldots, j_m\}$, we do a case analysis on which item of \Cref{def:equal_strict_ordering_function} is satisfied by $\delta_1,\ldots,\delta_m$.
\begin{itemize}
	\item Item 1: In such a case, we know that for all $j \in \{1,\ldots,m\}$, there exists $i \in \{n_a,\ldots,n_b\}$ such that $\delta_j(i) = \less$ implying that $\delta_j(i)$ is defined. Thus by hypothesis, $\tau'_j < \tau_i \leq \mathrm{max}(\tau_{n_a},\ldots,\tau_{n_b})$. Since for all $j \in \{1,\ldots, m\}$, $\tau'_j < \mathrm{max}(\tau_{n_a},\ldots,\tau_{n_b})$, we directly obtain that $\multiset{\tau'_1,\ldots,\tau'_m} <_m \multiset{\tau_{n_a},\ldots,\tau_{n_b}}$.
	\item Item 2: In such a case, we know that $\delta_1,\ldots, \delta_m$ are $(n_a,n_b)$-equal for $\{ j_1,\ldots, j_m\}$. Hence we already proved the steps $\tau_1,\ldots,\tau_m$ are smaller than $m$ distinct steps of the multiset $\multiset{\tau_{n_a},\ldots,\tau_{n_b}}$. When $m \leq n_b - n_a$, it directly entails that $\multiset{\tau'_1,\ldots,\tau'_m} <_m \multiset{\tau_{n_a},\ldots,\tau_{n_b}}$. When $m=n_b-n_a + 1$, the fact that there exists $k \in \{1,\ldots,m\}$ such that $\delta_{k}(j_{k}) = \less$ guarantees that at least one of the ordering functions is strict, and so $\tau'_k < \tau_{j_k}$. This allows us to conclude that $\multiset{\tau'_1,\ldots,\tau'_m} <_m \multiset{\tau_{n_a},\ldots,\tau_{n_b}}$.\qedhere
\end{itemize} 
\end{proof}

As previously mentioned, given a query or lemma of the form $\query = (F_1 \wedge \dots \wedge F_n \Rightarrow \qconcl)$, we will always assume the user-defined predicates occur in the premise after the standard predicates $\pevent$, $\patt$ and $\pmsg$. Hence for each query, we can associate an integer indicating the first index where a user-defined predicate occurs in the premise. By denoting $\indexPred{\query}$, we have that for all $1 \leq i < \indexPred{\query}$, $\pred(F_i) \in \{ \pevent, \patt, \pmsg\}$, and for all $\indexPred{\query} \leq i \leq n$, $\pred(F_i) \in \Fp$. When the premise does not contain user-defined predicates, we define $\indexPred{\query} = n+1$.

Similarly, as an ordered clauses $R = (H \rightarrow \bigcurlywedge_{i=1}^n F_i)$ originates from a  query, we can similarly define the first index where a user-defined predicate occurs in $\bigcurlywedge_{i=1}^n F_i$, denoted $\indexPred{R}$.


\begin{definition}
\label{def:strict-ordering}
Let $n_{idx}, n, m_{idx}, m \in \mathbb{N}$ such that $n_{idx} \leq n+1$ and $m_{idx} \leq m+1$. We say that the ordering functions $\delta_1,\ldots, \delta_m$ are strict for $(n_{idx},n)$ and $m_{idx}$ when there exist $\{ j_1,\ldots, j_{m_{idx}-1}\} \subseteq \{ 1,\ldots, n_{idx}-1\}$ and $\{ j_{m_{idx}}, \ldots, j_m\} \subseteq \{ n_{idx}, \ldots, n\}$ such that:
\begin{itemize}
	\item either $\delta_1,\ldots,\delta_{m_{idx}-1}$ are $(1,n_{idx}-1)$-strict for $\{ j_1,\ldots, j_{m_{idx}-1}\}$
	\item or $\delta_1,\ldots,\delta_{m_{idx}-1}$ are $(1,n_{idx}-1)$-equal for $\{ j_1,\ldots, j_{m_{idx}-1}\}$ and $\delta_{m_{idx}},\allowbreak\ldots,\delta_m$ are $(n_{idx},n)$-strict for $\{ j_{m_{idx}}, \ldots, j_m\}$
\end{itemize}
\end{definition}


\begin{lemma}
  \label{lem:ind}
Let $n_{idx}, n, m_{idx}, m \in \mathbb{N}$ such that $n_{idx} \leq n+1$ and $m_{idx} \leq m+1$. Let $\delta_1,\ldots, \delta_m$ be strict for $(n_{idx},n)$ and $m_{idx}$. Let $T$ be a trace.

For all tuples of steps $\tuplesteps = (\tau_1,\ldots, \tau_{n_{idx}-1})$, $\tuplesteps_p = (\tau_{n_{idx}},\ldots, \tau_n)$, $\tuplesteps' = (\tau'_1,\ldots, \tau'_{m_{idx}-1})$ and $\tuplesteps'_p = (\tau'_{m_{idx}},\ldots, \tau'_m)$, if for all $j \in \{1,\ldots,m\}$, for all $i \in \{1,\ldots, n\}$, $\delta_j(i)$ defined implies $\tau'_j \mathrel{\delta_j(i)} \tau_i$ then $(T,\tuplesteps',\tuplesteps'_p) \orderind (T,\tuplesteps,\tuplesteps_p)$
\end{lemma}

\begin{proof}
For $(T,\tuplesteps',\tuplesteps'_p) \orderind (T,\tuplesteps,\tuplesteps_p)$ to hold, we need to prove that either $\multiset{\tau'_1,\ldots,\allowbreak \tau'_{m_{idx}-1}} <_m \multiset{\tau_1,\ldots, \tau_{n_{idx}-1}}$ or else $\multiset{\tau_1,\ldots, \allowbreak\tau_{n_{idx}-1}} = \multiset{\tau'_1,\ldots, \tau'_{m_{idx}-1}}$ and $\multiset{\tau'_{m_{idx}},\ldots, \tau'_m} <_m \multiset{\tau_{n_{idx}},\ldots, \tau_n}$.

We do a case analysis on the item satisfied in \Cref{def:strict-ordering}. In Case 1, we know that $\delta_1,\ldots,\delta_{m_{idx}-1}$ are $(1,n_{idx}-1)$-strict for $\{ j_1,\ldots, j_{m_{idx}-1}\}$, hence we conclude by applying \Cref{lem:n_strict_ordering_function} to obtain $\multiset{\tau'_1,\ldots, \tau'_{m_{idx}-1}} <_m \multiset{\tau_1,\ldots, \tau_{n_{idx}-1}}$

In Case 2, we know that $\delta_1,\ldots,\delta_{m_{idx}-1}$ are $(1,n_{idx}-1)$-equal for $\{ j_1,\ldots, \allowbreak j_{m_{idx}-1}\}$. Hence by \Cref{lem:n_strict_ordering_function}, we deduce that $\multiset{\tau_1,\ldots, \tau_{n_{idx}-1}} \leq_m \multiset{\tau'_1,\ldots, \allowbreak\tau'_{m_{idx}-1}}$. If $\multiset{\tau_1,\ldots, \tau_{n_{idx}-1}} <_m \multiset{\tau'_1,\ldots, \tau'_{m_{idx}-1}}$ then we conclude; otherwise $\multiset{\tau_1,\ldots, \tau_{n_{idx}-1}} = \multiset{\tau'_1,\ldots, \tau'_{m_{idx}-1}}$. However, since we also know that $\delta_{m_{idx}},\allowbreak\ldots,\delta_m$ are $(n_{idx},n)$-strict for $\{ j_{m_{idx}}, \ldots, j_m\}$, then we can apply \Cref{lem:n_strict_ordering_function} to obtain $\multiset{\tau'_{m_{idx}},\ldots, \tau'_m} <_m \multiset{\tau_{n_{idx}},\ldots, \tau_n}$, which allows us to conclude.
\end{proof}

We can now state our extended transformation rule for applying inductive lemmas, denoted (\RindOrd{\Lemma_i,\allowedpreds}).

\begin{prooftree}
	\AxiomC{\begin{tabular}{c}
		$\Cset \cup \{ R = (H \rightarrow \bigcurlywedge_{i=1}^n F'_i)\}$\\
		$\query = (\bigwedge_{i=1}^m F_i \Rightarrow \bigvee_{j=1}^k \qconcl_j) \in \Lemma_i$\\
		$\query$ matches $R$ with $\delta_1,\ldots, \delta_m$ and $\sigma$\\
		$\delta_1,\ldots,\delta_m$ are strict for $(\indexPred{R},n)$ and $\indexPred{\query}$\\
		for all $i$, $\pred(F_i) \not\in \Fp$ implies $\delta \subsume_o \delta_i$
	\end{tabular}
	}
	\UnaryInfC{$\Cset \cup \{ H \wedge \ToB{\qconcl_k\sigma}_\delta \rightarrow \bigcurlywedge_{i=1}^n F'_i\}_{j=1}^k$}
\end{prooftree}


\paragraph{Soundness of the solving procedure}
The soundness result of the solving procedure is given by the following theorem that is an adaptation of \cite[Theorem 4]{BCC-sn22-long}.

\begin{theorem}
\label{th:solving_correspondence}
Let $\C_I$ be an initial instrumented configuration. Let $\allowedpreds$ be a set of predicates containing the predicate $\pevent$, all predicates in $\Fbp$ and all blocking predicates. Let $\Lemma$, $\Lemma_i$ be two sets of lemmas. Let $T \in \tracestepIO{\C_I}$. Let $\Cset$ be a set of simplified and selection free clauses Horn clauses consistent with $\Fp$ containing the selection free clauses of $\Cstd$.

For all ordered clauses $R$ satisfying $\allowedpreds$, for all ordered derivations $\D$ of $\bigcurlywedge_{i=1}^m F_i$ at steps $(\tau_1,\ldots,\tau_m)$ from $\{R\}$ and $\Cset \cup \CblockT{T}$ such that $\HypLemmas{T}{(\tau_1,\allowbreak\ldots,\tau_{\indexPred{R}-1})}{(\tau_{\indexPred{R}},\ldots,\tau_m)}{\Lemma}{\Lemma_i}$ and $T,\allowedpreds \satisfyD \D$, there exists an ordered derivation $\D'$ of $\bigcurlywedge_{i=1}^m F_i$ at steps $\tuplesteps$ from $\saturateS{\{R\}}{\Cset}{\allowedpreds}{\Lemma}{\Lemma_i}$ and $\Cset \cup \CblockT{T}$ such that $T,\allowedpreds \satisfyD \D'$.
\end{theorem}

\begin{proof}[Sketch proof]
The proof can be directly adapted from the one of \cite[Theorem 4]{BCC-sn22-long}. It also relies on showing how derivations can be transformed and use produce Horn clauses generated by the transformation rules. The main change comes the proof of soudness of the application of inductive lemmas which now relies on \Cref{lem:ind} instead of \cite[Lemma 14]{BCC-sn22-long}.
\end{proof}


\subsection{Verification of the correspondence query}

In~\cite[Section 7]{BCC-sn22-long}, a general algorithm is described to verify correspondence queries, including queries containing nested and injective queries. For the main query to solve, we keep these algorithms exactly as they are. The key difference comes from the way we prove attacker facts in a lemma. In a simple query, say $F \Rightarrow \patt(M)$, the semantics of ProVerif allows $M$ to be decucible from the set of attacker clauses, namely $\{(\ref{RInit}),(\ref{RGen}),(\ref{RFail}),(\ref{ruleRf})\}$ (see~\cite[Definition 30]{BCC-sn22-long}). Thus, $M$ may not appear in the attacker knowledge in the trace executing $F$. To apply our lemmas with attacker facts $\patt(M)$ in their conclusion, we need  $M$ to appear in the attacker knowledge before the step where $F$ is satisfied. Otherwise, we would break the soundness of lemma simplification rule. 

Thus, when proving a Lemma, instead of requesting that $\patt(M)$ is deducible from the restricted set of attacker clauses, we will request that $\patt(M)$ occurs directly in the hypothesis of the ordered clauses. Thanks to \Cref{enum-satisderiv-rules} of \Cref{def:satisfy_derivation}, this will ensure that $M$ is already in the attacker knowledge.

As an interesting consequence, ProVerif may fail to prove a
correspondence query when it is declared as a lemma but succeed when is is declared as a standard query. 

%% file: main-eurosp-2022.bbl
\begin{thebibliography}{10}
\providecommand{\url}[1]{#1}
\csname url@samestyle\endcsname
\providecommand{\newblock}{\relax}
\providecommand{\bibinfo}[2]{#2}
\providecommand{\BIBentrySTDinterwordspacing}{\spaceskip=0pt\relax}
\providecommand{\BIBentryALTinterwordstretchfactor}{4}
\providecommand{\BIBentryALTinterwordspacing}{\spaceskip=\fontdimen2\font plus
\BIBentryALTinterwordstretchfactor\fontdimen3\font minus
  \fontdimen4\font\relax}
\providecommand{\BIBforeignlanguage}[2]{{%
\expandafter\ifx\csname l@#1\endcsname\relax
\typeout{** WARNING: IEEEtran.bst: No hyphenation pattern has been}%
\typeout{** loaded for the language `#1'. Using the pattern for}%
\typeout{** the default language instead.}%
\else
\language=\csname l@#1\endcsname
\fi
#2}}
\providecommand{\BIBdecl}{\relax}
\BIBdecl

\bibitem{CT}
B.~Laurie, ``Certificate transparency,'' \emph{Communications of the ACM},
  vol.~57, no.~10, pp. 40--46, 2014.

\bibitem{rfc9162}
\BIBentryALTinterwordspacing
B.~Laurie, E.~Messeri, and R.~Stradling, ``Certificate transparency version
  2.0,'' Internet Requests for Comments, {RFC Editor}, {RFC} 9162, Dec. 2021.
  [Online]. Available: \url{http://www.rfc-editor.org/rfc/rfc1654.txt}
\BIBentrySTDinterwordspacing

\bibitem{Ryan:17}
M.~D. Ryan, ``Making decryption accountable,'' in \emph{Cambridge International
  Workshop on Security Protocols}.\hskip 1em plus 0.5em minus 0.4em\relax
  Springer, 2017, pp. 93--98.

\bibitem{armando2005avispa}
A.~Armando, D.~Basin, Y.~Boichut, Y.~Chevalier, L.~Compagna, J.~Cu{\'e}llar,
  P.~H. Drielsma, P.-C. H{\'e}am, O.~Kouchnarenko, J.~Mantovani \emph{et~al.},
  ``The avispa tool for the automated validation of internet security protocols
  and applications,'' in \emph{International conference on computer aided
  verification}.\hskip 1em plus 0.5em minus 0.4em\relax Springer, 2005, pp.
  281--285.

\bibitem{cheval2018deepsec}
V.~Cheval, S.~Kremer, and I.~Rakotonirina, ``Deepsec: deciding equivalence
  properties in security protocols theory and practice,'' in \emph{2018 IEEE
  Symposium on Security and Privacy (SP)}.\hskip 1em plus 0.5em minus
  0.4em\relax IEEE, 2018, pp. 529--546.

\bibitem{chadha2016automated}
R.~Chadha, V.~Cheval, {\c{S}}.~Ciob{\^a}c{\u{a}}, and S.~Kremer, ``Automated
  verification of equivalence properties of cryptographic protocols,''
  \emph{ACM Transactions on Computational Logic (TOCL)}, vol.~17, no.~4, pp.
  1--32, 2016.

\bibitem{meier2013tamarin}
S.~Meier, B.~Schmidt, C.~Cremers, and D.~Basin, ``The tamarin prover for the
  symbolic analysis of security protocols,'' in \emph{International conference
  on computer aided verification}.\hskip 1em plus 0.5em minus 0.4em\relax
  Springer, 2013, pp. 696--701.

\bibitem{BlanchetFnTPS16}
B.~Blanchet, ``Modeling and verifying security protocols with the applied pi
  calculus and {P}ro{V}erif,'' \emph{Foundations and Trends in Privacy and
  Security}, vol.~1, no. 1--2, pp. 1--135, Oct. 2016.

\bibitem{blanchet:hal-03366962}
B.~Blanchet, V.~Cheval, and V.~Cortier, ``Proverif with lemmas, induction, fast
  subsumption, and much more,'' in \emph{43rd {IEEE} Symposium on Security and
  Privacy, {SP} 2022, San Francisco, CA, USA, May 22-26, 2022}.\hskip 1em plus
  0.5em minus 0.4em\relax {IEEE}, 2022, pp. 69--86.

\bibitem{supplement}
``Proverif models and proverif source code,''
  \url{https://www.dropbox.com/sh/gbn5dy0amz1106f/AACbcILzg8o1Bhf5D3nMFa2Wa?dl=0},
  2022.

\bibitem{ARPKI}
D.~Basin, C.~Cremers, T.~H.-J. Kim, A.~Perrig, R.~Sasse, and P.~Szalachowski,
  ``{ARPKI}: Attack resilient public-key infrastructure,'' in \emph{Proceedings
  of the 2014 ACM SIGSAC Conference on Computer and Communications Security},
  2014, pp. 382--393.

\bibitem{DTKI}
J.~Yu, V.~Cheval, and M.~Ryan, ``{DTKI}: A new formalized pki with verifiable
  trusted parties,'' \emph{The Computer Journal}, vol.~59, no.~11, pp.
  1695--1713, 2016.

\bibitem{Ryan:ndss}
M.~D. Ryan, ``Enhanced certificate transparency and end-to-end encrypted
  mail,'' in \emph{NDSS Symposium}, 2014.

\bibitem{Kunemann:accountability}
R.~K{\"u}nnemann, I.~Esiyok, and M.~Backes, ``Automated verification of
  accountability in security protocols,'' in \emph{2019 IEEE 32nd Computer
  Security Foundations Symposium (CSF)}.\hskip 1em plus 0.5em minus 0.4em\relax
  IEEE, 2019, pp. 397--39\,716.

\bibitem{kroll2014secure}
J.~Kroll, E.~Felten, and D.~Boneh, ``Secure protocols for accountable warrant
  execution,'' 2014, \url{https://www.jkroll.com/papers/warrant_paper.pdf}.

\bibitem{nunez2019escrowed}
D.~Nu{\~n}ez, I.~Agudo, and J.~Lopez, ``Escrowed decryption protocols for
  lawful interception of encrypted data,'' \emph{IET Information Security},
  vol.~13, no.~5, pp. 498--507, 2019.

\bibitem{idan2020prshare}
L.~Idan and J.~Feigenbaum, ``Prshare: A framework for privacy-preserving,
  interorganizational data sharing,'' in \emph{Proceedings of the 19th Workshop
  on Privacy in the Electronic Society}, 2020, pp. 137--149.

\bibitem{9627586}
M.~Li, Y.~Chen, C.~Lal, M.~Conti, M.~Alazab, and D.~Hu, ``Eunomia: Anonymous
  and secure vehicular digital forensics based on blockchain,'' \emph{IEEE
  Transactions on Dependable and Secure Computing}, 2021.

\bibitem{pad.tech}
``Privacy-preserving accountable decryption,'' \url{https://pad.tech}, 2022.

\bibitem{kroll2015accountable}
J.~A. Kroll, ``Accountable algorithms,'' Ph.D. dissertation, Princeton
  University, 2015.

\bibitem{frankle2018practical}
J.~Frankle, S.~Park, D.~Shaar, S.~Goldwasser, and D.~Weitzner, ``Practical
  accountability of secret processes,'' in \emph{27th USENIX Security Symposium
  (USENIX Security 18)}, 2018, pp. 657--674.

\bibitem{proudler2014trusted}
G.~Proudler, L.~Chen, and C.~Dalton, ``Trusted platform architecture,'' in
  \emph{Trusted Computing Platforms}.\hskip 1em plus 0.5em minus 0.4em\relax
  Springer, 2014, pp. 109--129.

\bibitem{johnson2018titan}
S.~Johnson and D.~Rizzo, ``Titan silicon root of trust for google cloud,'' in
  \emph{Secure Enclaves Workshop}, 2018,
  \url{https://keystone-enclave.org/workshop-website-2018/slides/Scott_Google_Titan.pdf}.

\bibitem{opentitan}
``{OpenTitan}: an open source, transparent, high-quality reference design and
  integration guidelines for silicon root of trust (rot) chips,''
  \url{https://opentitan.org/}.

\bibitem{moller2021preliminary}
B.~H. M{\o}ller, J.~G. S{\o}ndergaard, K.~S. Jensen, M.~W. Pedersen, T.~W.
  B{\o}gedal, A.~Christensen, D.~B. Poulsen, K.~G. Larsen, R.~R. Hansen, T.~R.
  Jensen \emph{et~al.}, ``Preliminary security analysis, formalisation, and
  verification of opentitan secure boot code,'' in \emph{Nordic Conference on
  Secure IT Systems}.\hskip 1em plus 0.5em minus 0.4em\relax Springer, 2021,
  pp. 192--211.

\bibitem{pad-places-app}
``{PAD Places} - a location sharing app on ios and google implementing
  privacy-preserving accountable decryption,''
  \url{https://www.pad.tech/pad-places}, 2022.

\bibitem{parra2017kyc}
J.~Parra~Moyano and O.~Ross, ``{KYC} optimization using distributed ledger
  technology,'' \emph{Business \& Information Systems Engineering}, vol.~59,
  no.~6, pp. 411--423, 2017.

\bibitem{c4}
K.~Bhargavan, B.~Blanchet, and N.~Kobeissi, ``Verified models and reference
  implementations for the tls 1.3 standard candidate,'' in \emph{2017 IEEE
  Symposium on Security and Privacy (SP)}.\hskip 1em plus 0.5em minus
  0.4em\relax IEEE, 2017, pp. 483--502.

\bibitem{c25}
N.~Kobeissi, K.~Bhargavan, and B.~Blanchet, ``Automated verification for secure
  messaging protocols and their implementations: A symbolic and computational
  approach,'' in \emph{2017 IEEE European symposium on security and privacy
  (EuroS\&P)}.\hskip 1em plus 0.5em minus 0.4em\relax IEEE, 2017, pp. 435--450.

\bibitem{c26}
N.~Kobeissi, G.~Nicolas, and K.~Bhargavan, ``Noise explorer: Fully automated
  modeling and verification for arbitrary noise protocols,'' in \emph{2019 IEEE
  European Symposium on Security and Privacy (EuroS\&P)}.\hskip 1em plus 0.5em
  minus 0.4em\relax IEEE, 2019, pp. 356--370.

\bibitem{c9}
B.~Blanchet, ``Symbolic and computational mechanized verification of the
  arinc823 avionic protocols,'' in \emph{2017 IEEE 30th Computer Security
  Foundations Symposium (CSF)}.\hskip 1em plus 0.5em minus 0.4em\relax IEEE,
  2017, pp. 68--82.

\bibitem{c17}
V.~Cortier, D.~Galindo, and M.~Turuani, ``A formal analysis of the
  neuch{\^a}tel e-voting protocol,'' in \emph{2018 IEEE European Symposium on
  Security and Privacy (EuroS\&P)}.\hskip 1em plus 0.5em minus 0.4em\relax
  IEEE, 2018, pp. 430--442.

\bibitem{ProVerifManual}
B.~Blanchet, B.~Smyth, V.~Cheval, and M.~Sylvestre, ``Proverif 2.04: automatic
  cryptographic protocol verifier, user manual and tutorial,'' \emph{Version
  from}, 2021.

\bibitem{abadi06}
M.~Abadi and V.~Cortier, ``Deciding knowledge in security protocols under
  equational theories,'' \emph{Theoretical Computer Science}, vol. 367, no.~1,
  pp. 2--32, 2006.

\bibitem{cheval18}
V.~Cheval, V.~Cortier, and M.~Turuani, ``A little more conversation, a little
  less action, a lot more satisfaction: Global states in {ProVerif},'' in
  \emph{IEEE Computer Security Foundations Symposium (CSF)}, Jul. 2018, pp.
  344--358.

\bibitem{BCW-ccs22}
K.~Bhargavan, V.~Cheval, and C.~Wood, ``A symbolic analysis of privacy for tls
  1.3 with encrypted client hello,'' in \emph{{P}roceedings of the 29th {ACM}
  {C}onference on {C}omputer and {C}ommunications {S}ecurity ({CCS}'22)}.\hskip
  1em plus 0.5em minus 0.4em\relax Los Angeles, USA: ACM Press, Nov. 2022.

\bibitem{BCC-sn22-long}
B.~Blanchet, V.~Cheval, and V.~Cortier, ``Proverif with lemmas, induction, fast
  subsumption, and much more (long version),''
  \url{https://chevalvi.gitlabpages.inria.fr/chevalvi/files/BCC-snp22-long.pdf},
  Inria Paris, CNRS, LORIA, Tech. Rep., 2021.

\end{thebibliography}
